\documentclass{scrartcl}
\usepackage[utf8]{inputenc}
\usepackage{amsmath,amsfonts,amssymb}
\usepackage{amsthm}
\usepackage{bbold}
\usepackage[hidelinks]{hyperref}
\usepackage[backend=biber,style=numeric,sorting=none]{biblatex}
\usepackage{graphicx}
\usepackage{subcaption}
\usepackage{authblk}
\usepackage{xcolor}

\usepackage{biblatex}
\bibliography{bibliography}

\usepackage{appendix}

\DeclareMathOperator{\sources}{So}
\DeclareMathOperator{\sinks}{Si}
\DeclareMathOperator{\trophicDifferences}{TD}
\DeclareMathOperator{\fhierarchicalDifferences}{FHD}
\DeclareMathOperator{\bhierarchicalDifferences}{BHD}
\DeclareMathOperator{\mean}{Mean}
\DeclareMathOperator{\var}{Var}
\DeclareMathOperator{\argmin}{arg\,min}
\newcommand{\R}{\mathbb{R}}
\newcommand{\democracyCoefficient}[1]{\eta(#1)}
\newcommand{\fdemocracyCoefficient}[1]{\eta_f(#1)}
\newcommand{\bdemocracyCoefficient}[1]{\eta_b(#1)}
\newcommand{\fdemocracyCoefficientVertex}[2]{\eta_f(#1,#2)}
\newcommand{\bdemocracyCoefficientVertex}[2]{\eta_b(#1,#2)}
\newcommand{\fhierarchicalincoherence}[1]{\rho_f(#1)}
\newcommand{\bhierarchicalincoherence}[1]{\rho_b(#1)}

\theoremstyle{definition}
\newtheorem{theorem}{Theorem}[section] 
\newtheorem{lemma}[theorem]{Lemma}
\newtheorem{definition}[theorem]{Definition}
\newtheorem{conjecture}[theorem]{Conjecture}
\newtheorem{corollary}[theorem]{Corollary}

\begin{document}
 
\title{Graph Hierarchy}
\subtitle{A novel approach to understanding hierarchical structures in complex networks}
\date{July 2020}

\author[1]{Giannis~Moutsinas}
\author[2]{Choudhry~Shuaib}
\author[3]{Weisi~Guo}
\author[4]{Stephen~Jarvis}

\affil[1]{\footnotesize School of Computing, Electronics and Mathematics, Coventry University, Coventry, UK}
\affil[2]{\footnotesize Department of Computer Science, University of Warwick, Coventry, UK}
\affil[3]{\footnotesize Centre for Autonomous and Cyberphysical Systems, Cranfield University, Cranfield, UK}
\affil[4]{\footnotesize College of Engineering and Physical Sciences, University of Birmingham, Birmingham, UK}

\maketitle

\begin{abstract}
Trophic coherence, a measure of a graph's hierarchical organisation, has been shown to be linked to a graph's structural and dynamical aspects such as cyclicity, stability and normality.
Trophic levels of vertices can reveal their functional properties and partition and rank the vertices accordingly.
Yet trophic levels and hence trophic coherence can only be defined on graphs with basal vertices, vertices with zero in-degree.
Consequently, trophic analysis of graphs had been restricted until now. In this paper we introduce a novel framework, a generalisation of trophic levels, which we call hierarchical levels, that can be defined on any simple graph.
Within this general framework, we develop additional metrics named influence centrality, a measure of a vertices ability to influence dynamics, and democracy coefficient, a measure of overall feedback in the system, both of which have implications for the controllability of complex systems.
We discuss how our generalisation relates to previous attempts and what new insights are illuminated on the topological and dynamical aspects of graphs.
Finally, we show how the hierarchical structure of a network relates to the incidence rate in a SIS epidemic model. 
\end{abstract}

\newpage

\section{Introduction}

Patient zero is the start of an epidemic that spreads through a city. A rumour spreads like wildfire amongst a group of friends. An accident happens on the road and the associated disturbance spreads congestion throughout the road network in the vicinity of the incident. These are just a small number of examples of real life processes involving the directed flow of some quantity, whether it be information or physical, across a graph structure. Graphs are omnipresent and they constitute many of the complex systems that underlie much of our infrastructure and social interactions as well as ecological and biological systems that control and regulate life. Since the turn of the millennium there has been an explosion of research in network science\cite{newman2003structure}, \cite{watts1998collective}. Understanding how signals or processes percolate through a graph and what role topology and structure play, has been a key research aim \cite{barrat2008dynamical}. 

Hierarchical structure is pervasive across complex networks with examples spanning from neuroscience\cite{bassett2017network}, economics\cite{antras2012measuring}, social organisations\cite{krackhardt2014graph}, urban systems\cite{batty1994fractal}, communications\cite{vazquez2002large}, pharmaceuticals\cite{csermely2013structure} and biology\cite{cheng2015approach},  particularly metabolic\cite{ravasz2002hierarchical} and gene networks\cite{gerstein2012architecture}. Previous work such as \cite{ravasz2003hierarchical}, \cite{crofts2011googling}, \cite{corominas2013origins}, \cite{mones2012hierarchy}, \cite{trusina2004hierarchy} and \cite{coscia2018using} attempt to further advance the study of hierarchical structures in complex graphs from different perspectives. These approaches typically cast hierarchy into a dichotomy; hierarchy in terms of a modular organisation known as nested hierarchy, an instance of community structure, or alternatively, a directed flow hierarchy. Graph Hierarchy covers both the order and flow hierarchy approaches. We will primarily focus on the second of these approaches but the framework we expound in this paper is applicable in the undirected scenario. Our work differs from these previous attempts by taking the direction of generalising the trophic approach consequently producing a richer framework allowing for a deeper understanding of the hierarchy phenomenon underlying many real world complex systems. Subsequently, we have a more general and elegant architecture for studying hierarchical structures in complex graphs with a rich history in ecology \cite{dunne2004network} that can be built upon and referenced when studying complex networks in this approach.

Graphs play a significant role in ecology \cite{ings2009ecological} where graph tools are used to understand the complex ecosystems and food webs that are present in our environment. Ecological networks are directed graphs representing biological interactions and exhibit a natural trophic structure. Researchers have defined a quantity known as trophic level to illustrate the hierarchical nature \cite{levine1980several} of these graphs. Trophic coherence provides a measure of organisation via the distribution of differences of trophic level among the vertices of the graph which are adjacent to one another. This describes how neatly the structure is defined by discrete levels or partitions in a directed hierarchy. Research has shown that trophic coherence is a proxy for the stability of a food web \cite{johnson2014trophic} and found that the lack of cycles in a graph is inherently linked with the trophic coherence of a graph \cite{johnson2017looplessness}. These ideas have been used to analyse the spread of infections on a graph \cite{klaise2016neurons}, and to assess robustness and resilience of rail networks \cite{pagani2019resilience}. This illustrates a link between the stability and dynamics of graph processes, and the underlying structure of the graph. Trophic levels have only been defined for graphs where there are clear basal vertices, i.e. vertices with zero in-degree. In this paper we define a notion of trophic levels which is applicable to any graph. Using our definition we can apply a trophic analysis to graphs of any type and determine a hierarchy of the vertices present in the graph. Moreover, we introduce the hierarchical incoherence parameter as a measure of alignment of direction and organisation in a graph structure. The hierarchical incoherence parameter corresponds to the trophic incoherence parameter defined by Johnson et al in \cite{johnson2014trophic}. 

We introduce an additional novel vertex metric, influence centrality, a measure of a vertex's ability to influence the dynamics. This is crucial because forward influencer vertices drive the dynamics of a directed graph. In ecology, these are the basal species at the bottom of the food web; in epidemiology, this is the outbreak zone; and in morning transportation, the commuter towns. This is intimately tied with the study of the controllability of complex systems \cite{liu2011controllability}. When identifying the driver vertices, typically, expensive maximum matching algorithms are utilised. Graph Hierarchy provides an efficient algorithm with additional information. We also introduce the democracy coefficient as a measure of the size of influence subgraphs, a control kernel if you will, that is not influenced by the rest of the graph; understood as a measure of feedback present in the system. We then go onto show that the democracy coefficient correlates strongly with the topology of the graph and its corresponding relationship with the hierarchical incoherence parameter. Finally, we study the diffusive properties of a graph in this framework by modelling a contagion dynamics.

\section{Preliminaries}

Trophic levels were devised in order to model energy flow between species in a food web. In this context a vertex with no in-neighbours represents a primary producer species, for example grass. These are the species that provide energy into the food web.
In food webs such vertices are called basal, in this article we will follow the convention of flow networks and we will call such vertex a \textit{source}. Similarly a vertex with no out-neighbours will be called a \textit{sink}. 

Throughout this article we will consider only weighted simple directed graphs and we will use $G$ to denote them. We will denote by $G^T$ the transpose graph of $G$, i.e. the graph we get if we reverse the direction of all edges in $G$. We will denote by $\sources(G)$ and $\sinks(G)$ the set of all source and sink vertices respectively. Moreover, if $H$ is a subgraph of $G$, we will denote by $G\setminus H$ the graph that we get by removing from $G$ all vertices in $H$ with all the edges that are adjacent to them.
We will call the ``graph'' $G\setminus G$ the \textit{empty graph}.

We will shift our point of view from energy flow to information flow and we consider the following dynamics on a weighted graph. We assign a colour to each vertex. Then, at each time step a vertex chooses at random a vertex between itself and its in-neighbours with a probability  proportional to the weighted in-degree. Its own weight is always 1. Then all vertices update their colour to the colour of their chosen vertex simultaneously. We will call this \textit{forward influence dynamics}. In this dynamics, source vertices are important because they stay at their original colour forever.

We define \textit{backward influence dynamics} to be the same process on the transpose graph, i.e. the colour of a vertex is updated by considering the out-neighbours instead of the in-neighbours.
Similarly to the forward influence dynamics, a sink vertex will remain at its original colour for ever.


\begin{definition}
A graph $G$ is called  \textit{simply forward influenced} if $\sources(G)$ and $G\setminus\sources(G)$ are not empty and for any $v\in G \setminus \sources(G)$ there exists $u\in \sources (G)$ such that there is a directed path from $u$ to $v$. Similarly, $G$ will be called \textit{simply backward influenced} if its transpose graph $G^T$ is \textit{simply forward influenced}.
\label{definition_simply_influenced}
\end{definition}

Food webs are typically simply forward influenced. We prove in Appendix \ref{sec_proofs} that trophic levels can be defined on a graph if and only if it is simply forward influenced.

Sources in a graph do not have to be single vertices, they can also be subgraphs. This leads us to the following definition.

\begin{definition}
Let $G$ be a weakly connected graph and $\Gamma$ a subgraph of $G$. Then $\Gamma$ is called a \textit{minimal source subgraph} of $G$ if there is no edge from $G\setminus\Gamma$ to $\Gamma$ and this property fails if we remove any vertex from $\Gamma$. A subgraph $\Gamma$ is called \textit{minimal sink subgraph} if it is a minimal source subgraph for the transpose graph, $G^T$.
\end{definition}

\begin{definition}
For a weakly connected graph $G$, let $\Gamma_1$, $\dots$, $\Gamma_l$ be its minimal source subgraphs and let $\Delta_1$, $\dots$, $\Delta_l$ be its minimal sink subgraphs of $G$.

The subgraph, $F_f$ (resp. $F_b$), that we get if we remove all edges that belong to $\Gamma_i$'s (resp. $\Delta_i$'s) from $G$ and delete all isolated vertices is called the \textit{simply forward} (resp.  \textit{backward}) \textit{influenced subgraph} of $G$.

We call the set $\{\Gamma_1$, $\dots$, $\Gamma_l$, $F_f\}$ (resp.  $\{ \Delta_1$, $\dots$, $\Delta_l$, $F_b\}$) the \textit{forward} (resp.  \textit{backward})  \textit{hierarchical decomposition} of $G$.

The subgraph, $C$, that we get if we remove all edges that belong to $\Gamma_i$'s and $\Delta_i$'s from $G$ and delete all isolated vertices is called \textit{core subgraph} of $G$ and the set $\{\Gamma_1$, $\dots$, $\Gamma_l$, $C$, $\Delta_1$, $\dots$, $\Delta_l\}$ is called the \textit{hierarchical decomposition} of $G$

If the hierarchical decomposition of $G$ is not trivial, i.e. $\Gamma_1=\Delta_1=G$, we call $G$, hierarchically deccomposable.
\end{definition}

A graph $G$ is strongly connected then its hierarchical decompostion is trivial. We prove in Appendix \ref{sec_proofs} that a graph is not hierarchically decomposable if and only if it is strongly connected.
Notice that if $G$ is hierarchically decomposable, then $G^T$ is also hierarchically decomposable.
An example of a hierarchically decomposable graph, its maximal simply influenced subgraphs and its core subgraph can be seen in Figure \ref{fig_hierarchically_decomposable_graph}.

\begin{figure}[t]
    \centering
    \begin{subfigure}[t]{0.24\textwidth}
        \includegraphics[height=1.9in]{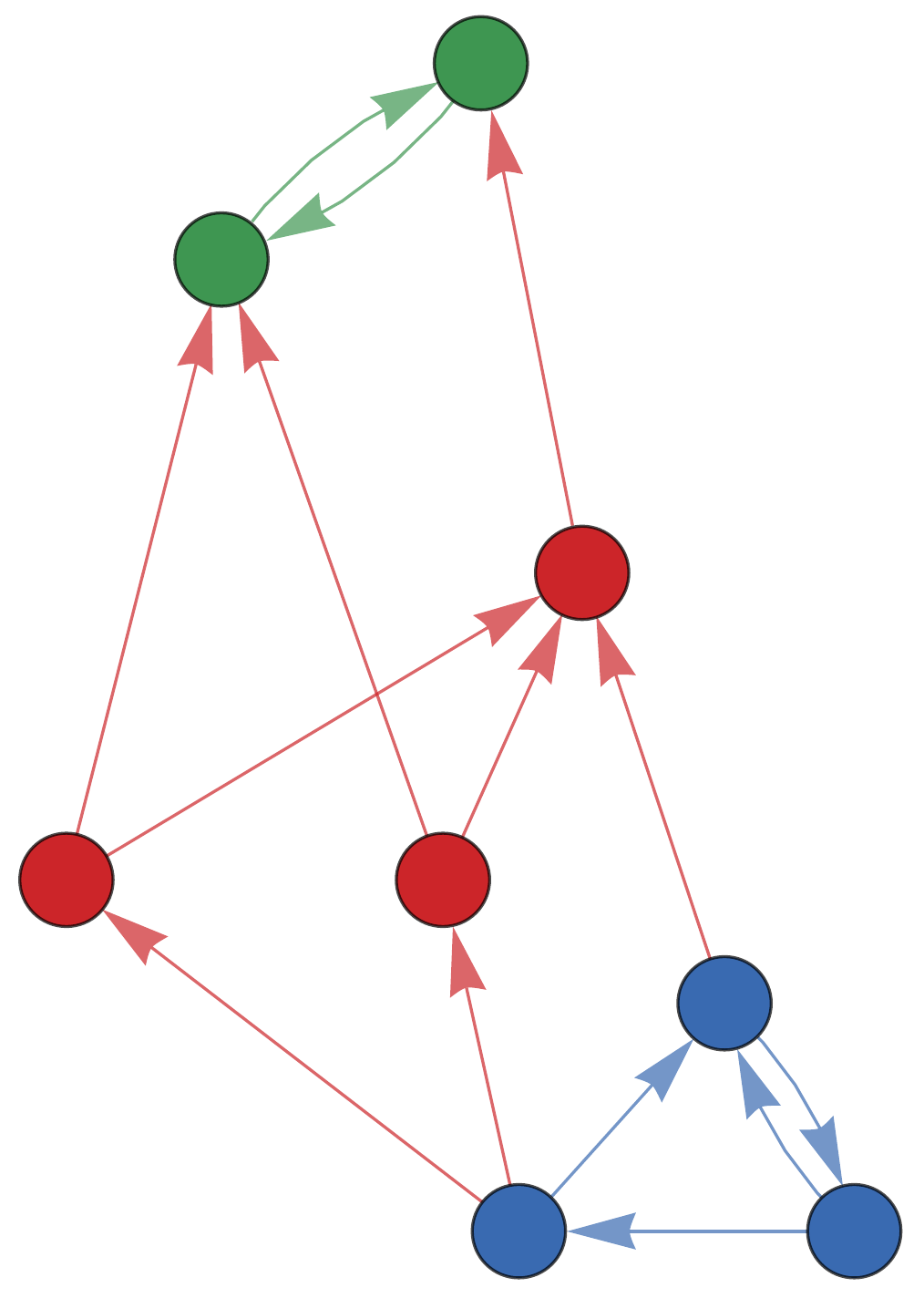}
        \subcaption{}
    \label{fig_hierarchically_decomposable_graph_example}
    \end{subfigure}~
    \begin{subfigure}[t]{0.24\textwidth}
        \includegraphics[height=1.9in]{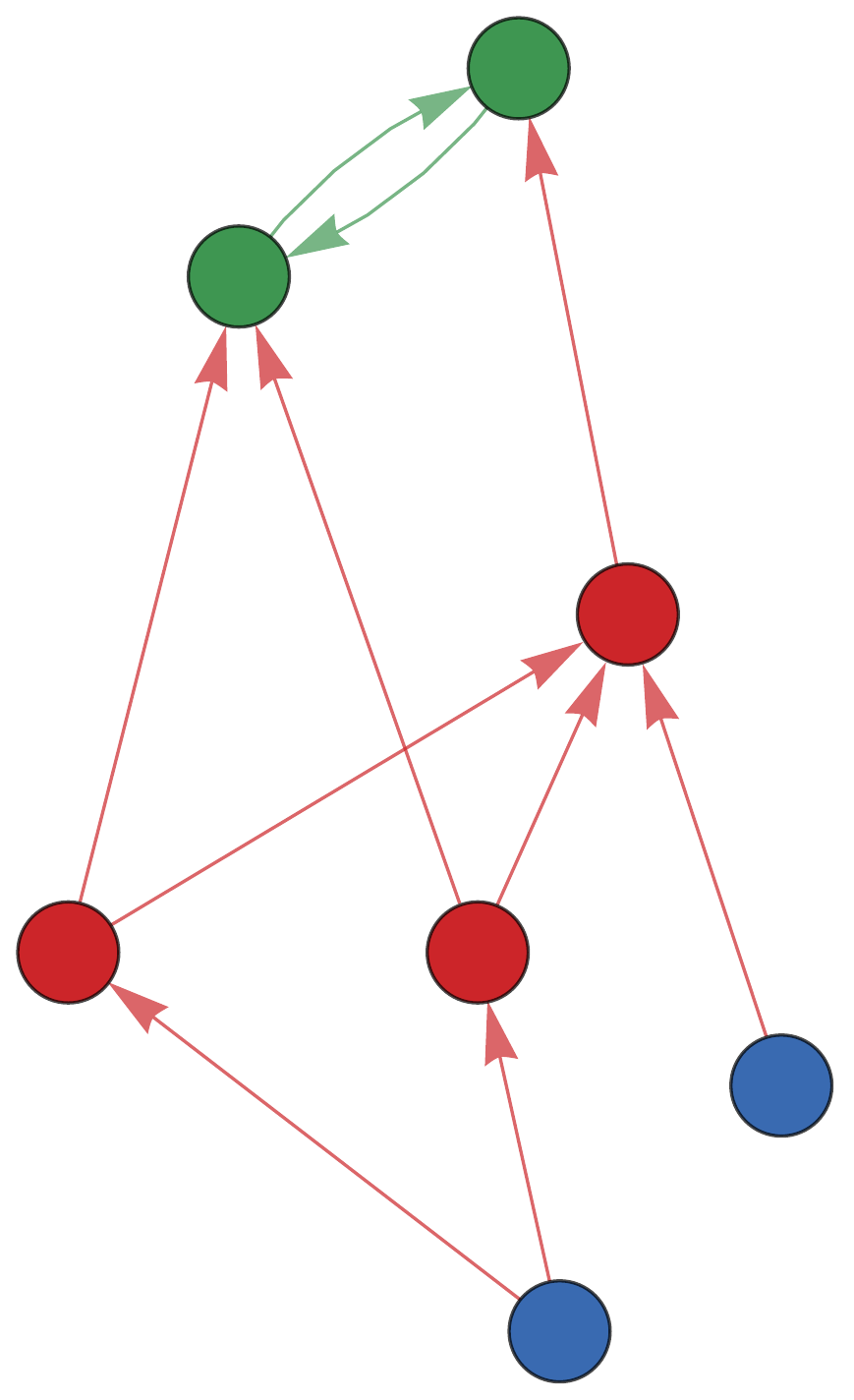}
        \subcaption{}
    \label{fig_hierarchically_decomposable_graph_forward_subgraph}
    \end{subfigure}~
    \begin{subfigure}[t]{0.24\textwidth}
        \includegraphics[height=1.9in]{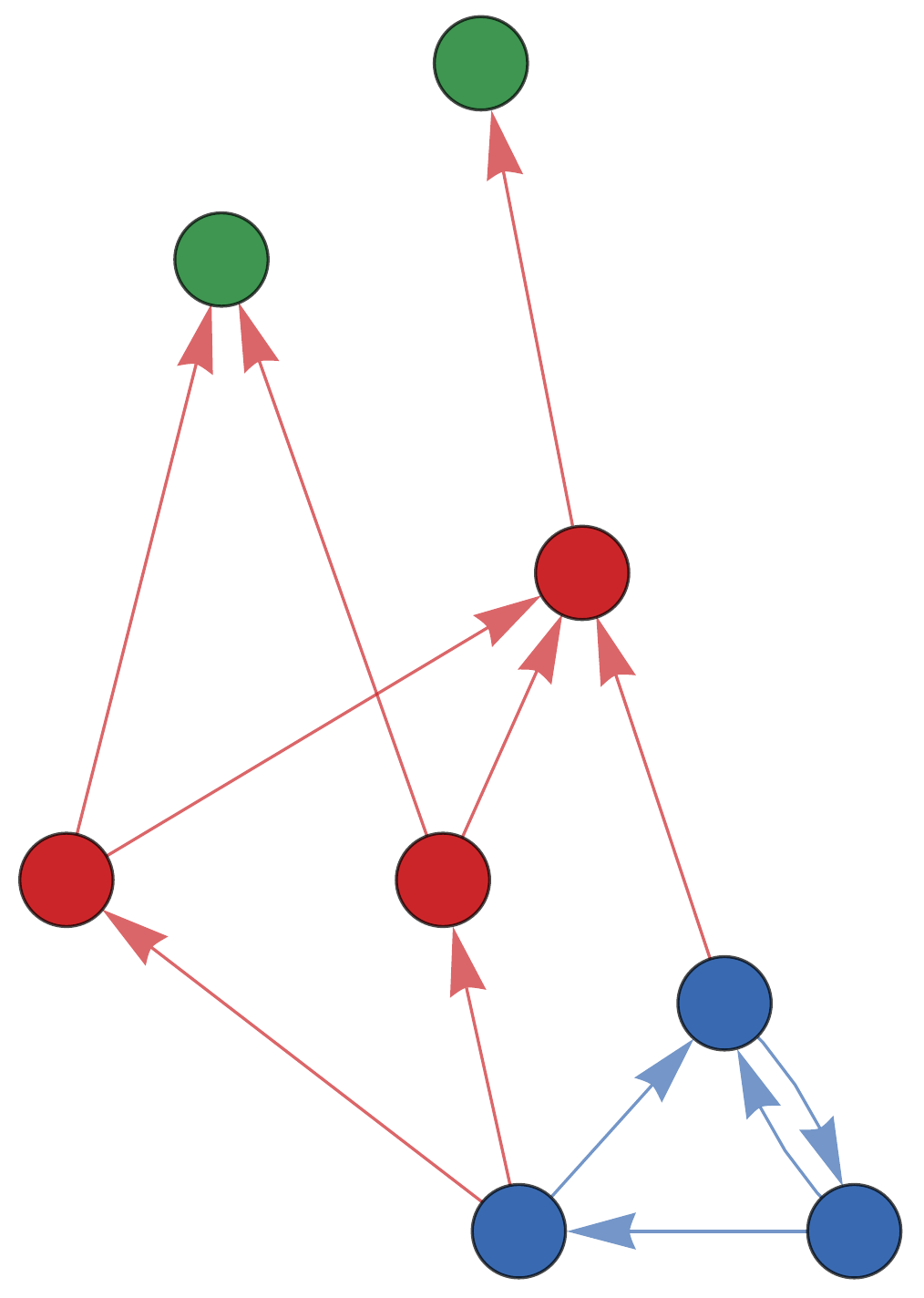}
        \subcaption{}
    \label{fig_hierarchically_decomposable_graph_backward_subgraph}
    \end{subfigure}~
    \begin{subfigure}[t]{0.24\textwidth}
        \includegraphics[height=1.9in]{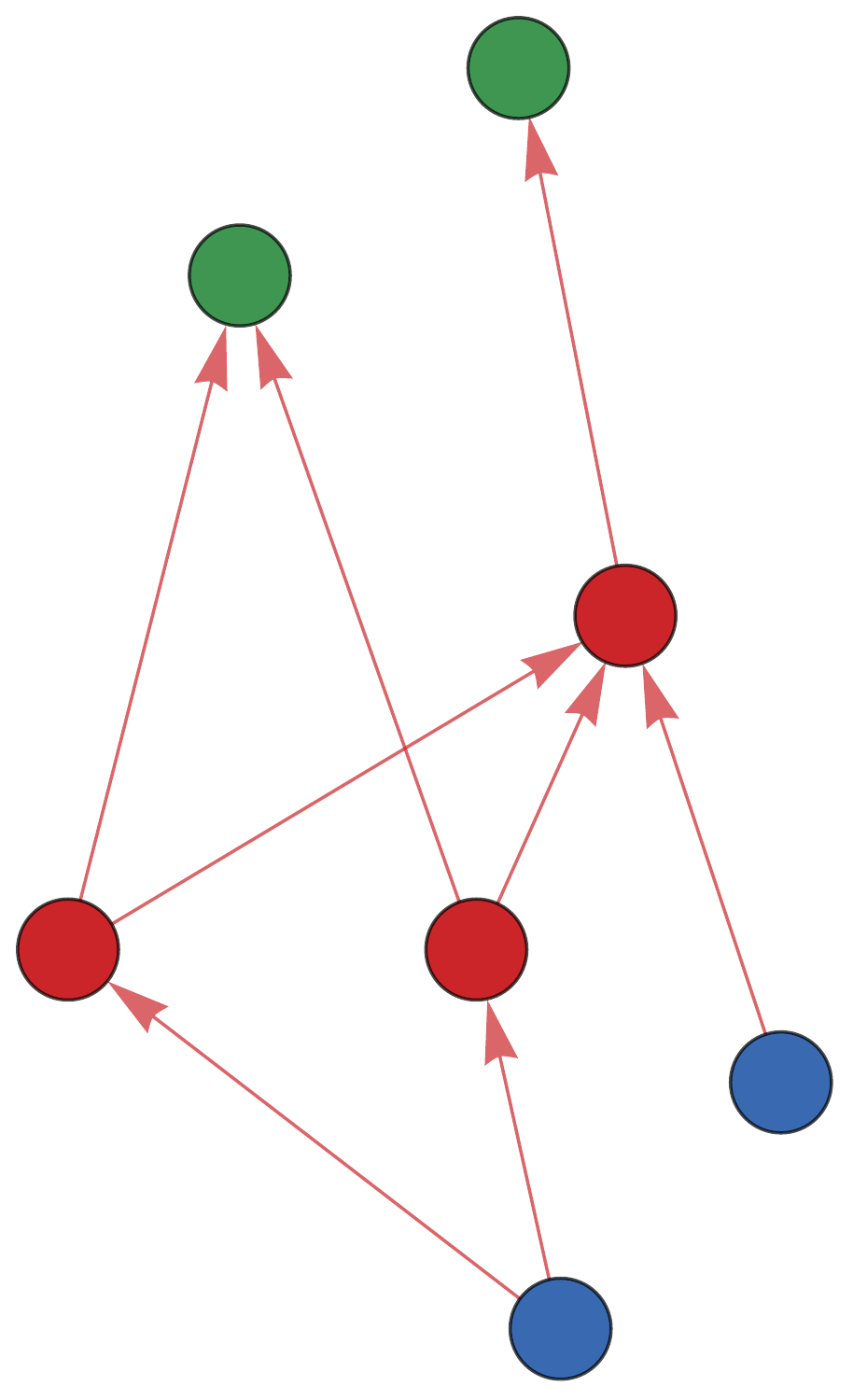}
        \subcaption{}
    \label{fig_hierarchically_decomposable_graph_core_subgraph}
    \end{subfigure}
    \caption{An example of a hierarchically decomposable graph. (a) The original graph. Its source subgraph is marked blue and its sink subgraph is marked green. If we apply the forward influence dynamics on the graph, then in finite time all vertices will become blue. On the other hand if we apply the backward influence dynamics, then in finite time all vertices will become green. In both case, red vertices do not affect the asymptotic state. (b) The simply forward influenced subgraph. (c) The simply backward influenced subgraph. (d) The core subgraph.}
    \label{fig_hierarchically_decomposable_graph}
\end{figure}

The hierarchical decomposition of a graph can be compared to the bow-tie structure of graphs \cite{wright2019central}. An input, which contains the source vertices; an output, which contains sink vertices. Finally, a core which is comprised of vertices that cannot reach the input and cannot be reached by the output. The difference is that in our case we do not impose the restriction that the core has to be strongly connected. 

\begin{definition}
Let $G$ be a hierarchically decomposable graph and $H_1$, $\dots$, $H_l$ be its minimal source (resp. sink) subgraphs, a vertex $v\in G$ is called a \textit{forward} (resp. \textit{backward}) \textit{influencer} if $v\in \cup_i H_i$ and $v$ is called \textit{forward} (resp. \textit{backward}) \textit{influenced} if $v\in G\setminus\cup_i H_i$.
\end{definition}

\subsection*{Notation}

We will use $A$ to represent the weighted adjacency matrix of a simple positively weighted graph and we will denote its entries by $a_{ij}$. If there is no directed edge between $i$ and $j$ then we have $a_{ij}=0$ and if there is a directed edge from $i$ to $j$, then $a_{ij}$ is positive and represents the edge's weight. 

We define $d_i=\sum_{j}a_{ji}$ to be the weighted in-degree of vertex $i$, $d=(d_1,\dots,d_n)$ the weighted in-degree vector.
The weighted in-degree Laplacian of a graph is defined to be the matrix $L=\text{diag}(d)-A$.
For notational convenience we define $M=L^\mathsf{T}$, where $L^\mathsf{T}$ is the transpose of $L$.
Similarly we define $\delta$ to be the weighted out-degree vector and the out-degree Laplacian of a graph is the matrix $\Lambda = \text{diag}(\delta) - A$.

\section{Hierarchical Levels and Hierarchical Differences}

In this section we introduce several, local and global, graph metrics. The notions of hierarchical levels and influence centrality are local metrics that can be used as vertex centrality/ranking measures. Hierarchical levels are a direct generalisation of trophic levels.
We discuss the relation between trophic levels and hierarchical levels further in Appendix \ref{sec_trophic}.
Hierarchical levels can be computed on any positively weighted simple directed graph\footnote{If a graph is undirected, it can be turned into a directed graph by replacing each undirected edge by a pair of directed edges.}.
The notions of democracy coefficient and hierarchical incoherence are global metrics that characterise the graph.
The notions of influence centrality and democracy coefficient have explicit connections with the topology of the graph.
In both cases we define two versions of the metrics, forward and backward, capturing the different notions of control and dependence.
The backward metrics are equivalent to the forward version applied to the transposed graph.

\subsection{Hierarchical levels}
\label{sec_HL_definition}

We now define the notions of hierarchical levels. Forward hierarchical levels are defined through the matrix $M$ and backward hierarchical levels are defined through the matrix $\Lambda$. We will call the difference of forward and backward hierarchical levels simply hierarchical levels.

\begin{definition}
\label{definition_FHL_non_trophic_graphs}

Let $G$ be a directed graph, $d$ be its in-degree vector, $L$ be its in-degree Laplacian matrix and $M=L^\mathsf{T}$.
The vector of \textit{forward hierarchical levels}, $g$, is
\begin{equation*}
g:=\underset{x\in \mathcal T}\argmin\|x\|_2 \text{, where } \mathcal T:=\underset{x\in\R^n}\argmin\|M x-d\|_2.
\end{equation*}
Similarly the vector of \textit{backward hierarchical levels}, $\gamma$, is
\begin{equation*}
\gamma:=\underset{x\in \mathcal S}\argmin\|x\|_2 \text{, where } \mathcal S:=\underset{x\in\R^n}\argmin\|\Lambda x-\delta\|_2.
\end{equation*}
Finally we define the vector of \textit{hierarchical levels} of $G$ to be $h=\frac{1}{2}(g-\gamma)$.
\end{definition}

We note that by this definition that we obtain $g=M^+ d$ and $\gamma = L^+\delta$, where $M^+$ denotes the Moore-Penrose inverse of the matrix $M$, see \cite{penrose1956best}. Intuitively forward hierarchical levels grade the vertices based on their distance from source subgraphs. Backward hierarchical levels rank the vertices based on their distance from sink subgraphs. They provide the perspectives of control and dependence respectively. Hierarchical levels take into account both perspectives, to be utilised when drawing a graph in order to emphasise and reveal its overall hierarchical structure.

The computation of the varying hierarchical levels boils down to a convex optimisation problem, for which we have reliable and efficient, polynomial time, algorithms available \cite{nesterov1994interior}. By definition it is easy to see that the forward (resp.  backward) hierarchical levels of a graph $G$ are equal to the backward (resp. forward hierarchical) levels of $G^T$, respectively. This means that the vector of hierarchical levels of $G^T$ is $-h$, where $h$ the vector of hierarchical levels of $G$.

\begin{figure}[t]
    \centering
    \begin{subfigure}[t]{0.22\textwidth}
    \includegraphics[width=\textwidth]{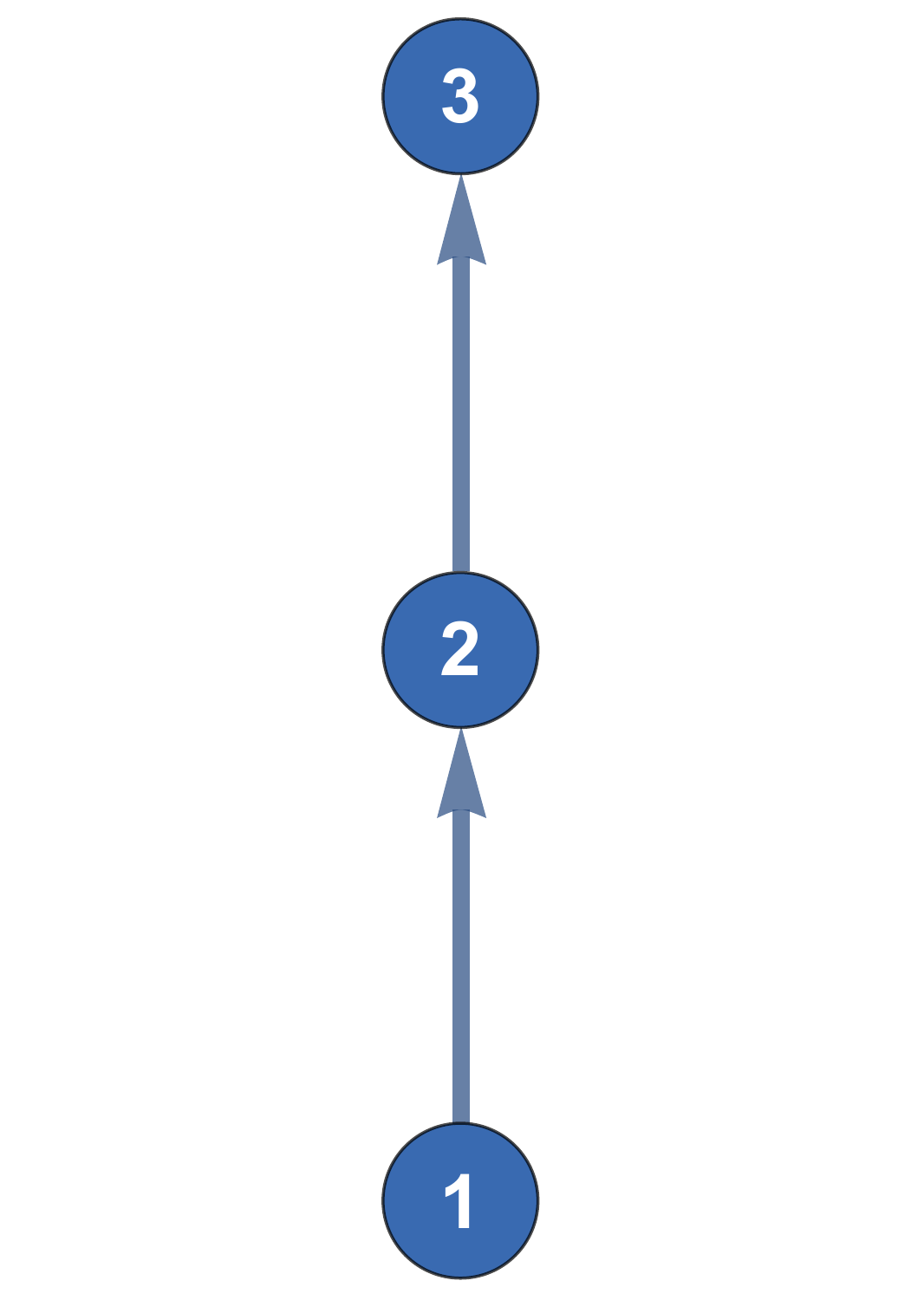}
    \subcaption{}
    \label{fig_HL_example_1}
    \end{subfigure}~
    \begin{subfigure}[t]{0.22\textwidth}
    \includegraphics[width=\textwidth]{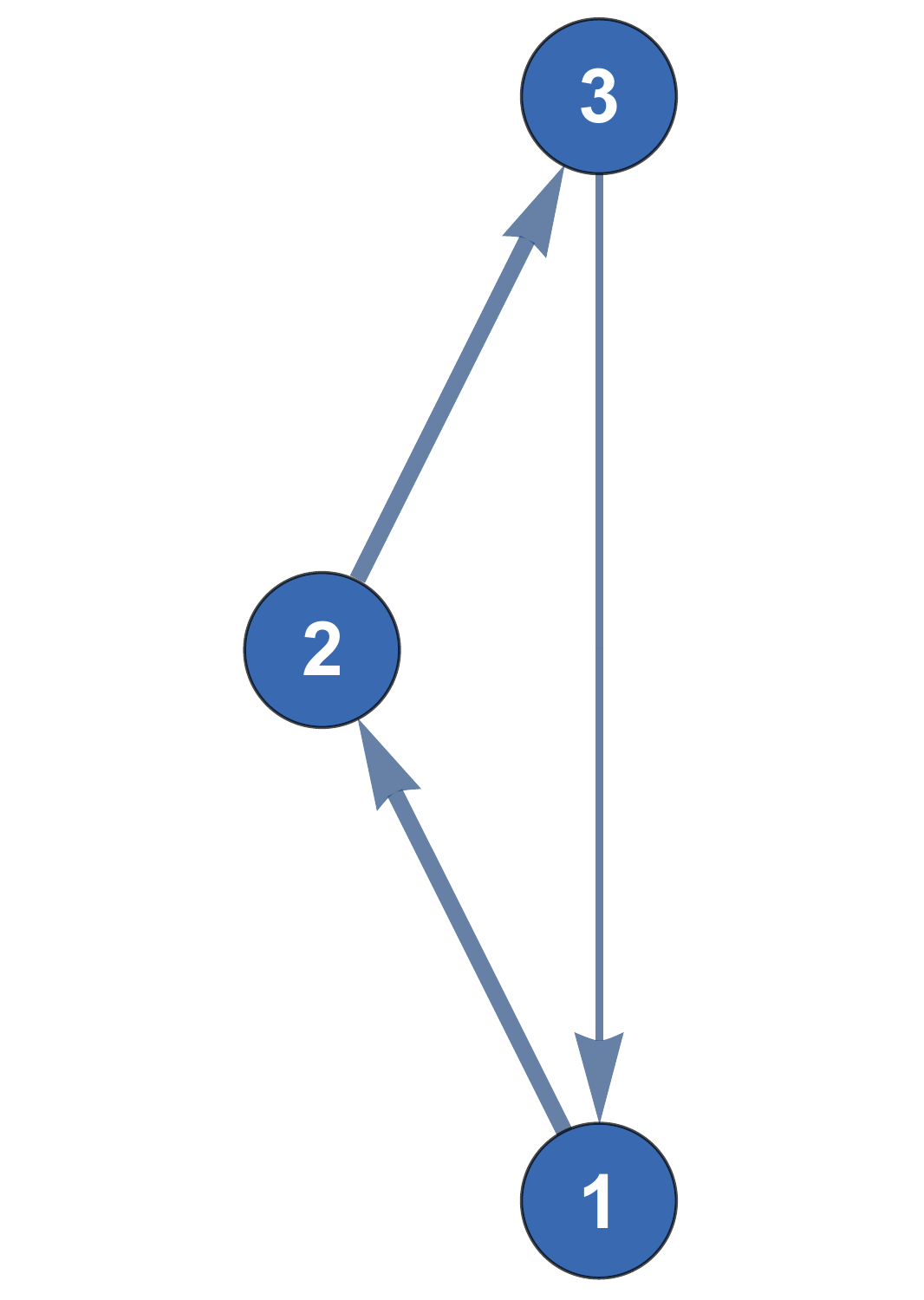}
    \subcaption{}
    \label{fig_HL_example_2}
    \end{subfigure}~
    \begin{subfigure}[t]{0.22\textwidth}
    \includegraphics[width=\textwidth]{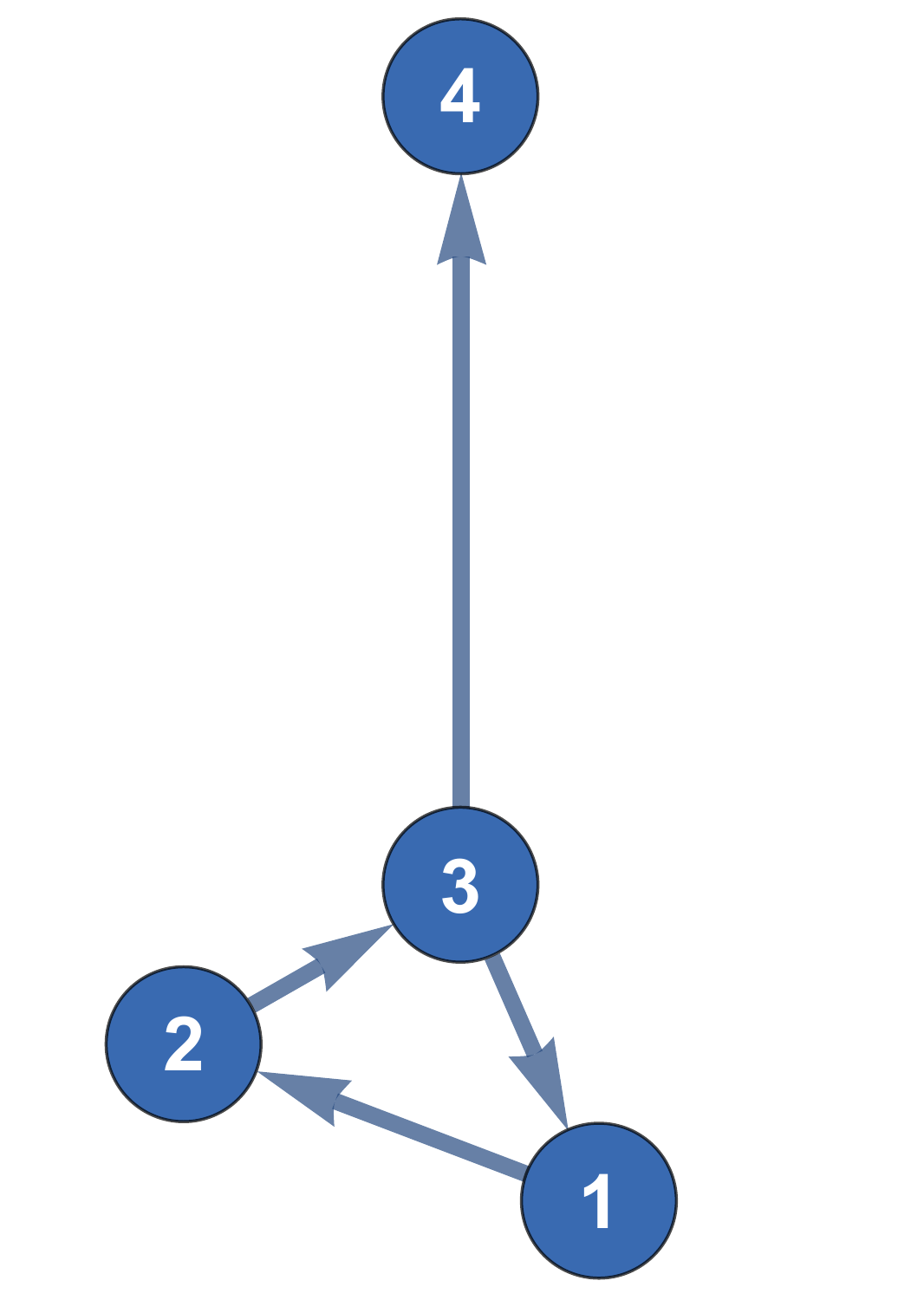}
    \subcaption{}
    \label{fig_HL_example_3}
    \end{subfigure}
    \begin{subfigure}[t]{0.22\textwidth}
    \includegraphics[width=\textwidth]{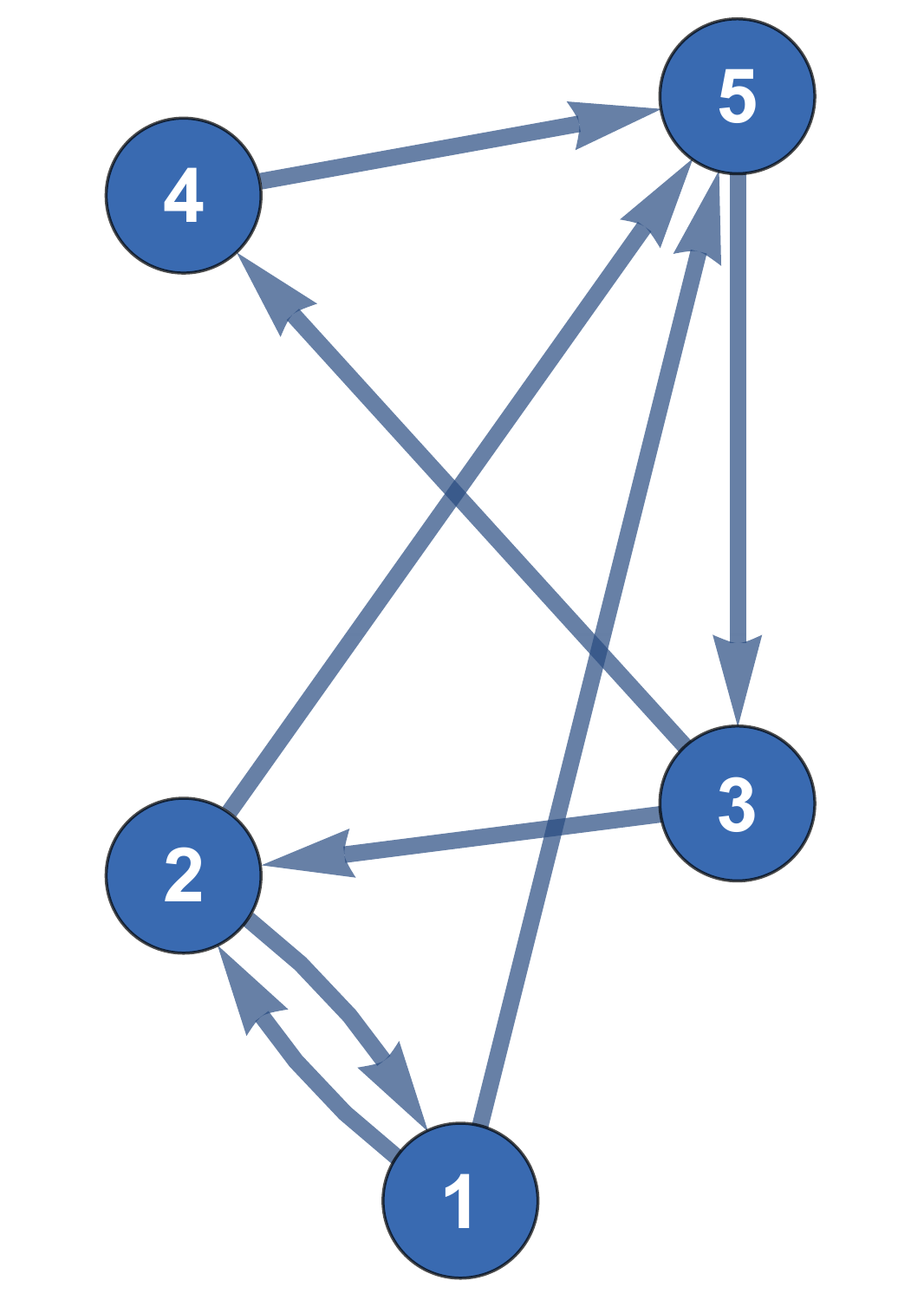}
    \subcaption{}
    \label{fig_HL_example_4}
    \end{subfigure}
    \caption{Examples of hierarchical layout of graphs: The $y$ coordinate of a vertex corresponds to its hierarchical level. In (b) the edge $3\to1$ has smaller weight that the other two edges.}
    \label{fig_HL_examples}
\end{figure}

Some simple examples are illustrated in Figure \ref{fig_HL_examples}. The $y$ coordinate of a vertex corresponds to its hierarchical level. In Figure \ref{fig_HL_example_1} a directed 3-chain is shown. The forward and backward hierarchical levels of the vertices are $(-1,0,1)$ and $(1,0,-1)$ respectively. The hierarchical levels are $(-1,0,1)$, which agrees with our intuition.

In the case of a directed 3-cycle, because of its symmetry we expect that all vertices will have the same hierarchical level, 0. We break the symmetry of the 3-cycle in two ways. In Figure \ref{fig_HL_example_2} we set the weight of the edge $3\to1$ to be $1/2$ and the rest of the weights equal to 1. This is halfway between a 3-chain and a 3-cycle. Intuitively we expect that in this case vertex 1 should be lower than the rest. We find that the forward and backward hierarchical levels of the vertices are $(-0.5,0,0.5)$ and $(0.5,0,-0.5)$ respectively, which implies that the hierarchical levels are $(-0.5,0,0.5)$.

The other way we break the 3-cycle's symmetry is by adding another vertex, Figure~\ref{fig_HL_example_3}. In this case we will see that the hierarchical differences of the 3-cycle do not change because of the additional vertex, only the hierarchical levels change. The forward hierarchical levels are $(-0.25,-0.25,-0.25,0.75)$. However when we look at the backward hierarchical levels, we see that the vertices of the 3-cycle are not equivalent. The backward hierarchical levels are $(2.25, 1.25, 0.25, -3.75)$, so the hierarchical levels are $(-1.25, -0.75, -0.25, 2.25)$. 

The final example, Figure \ref{fig_HL_example_4}, is a strongly connected graph. Intuitively we expect that in a strongly connected graph, vertices with low in-degree and high out-degree will tend to have low hierarchical levels.
However, the hierarchical level of a vertex is not a local property and it depends also on the closed paths that start and finish on the vertex.
The forward hierarchical levels are $(-0.783, -0.0333, -0.117, 0.3, 0.633)$ and the backward hierarchical levels are $(0.725, 0.607, 0.336, -0.868, -0.8)$. The hierarchical levels of this graph are $(-0.754, -0.32, -0.226, 0.584, 0.717)$ approximately.

Additionally because our measure can be applied to strongly connected graphs, this means it's applicable to undirected graphs but with a caveat. In the case of undirected graphs the adjacency matrix is symmetric and so the forward and backward hierarchical levels would be equal to each other. Consequently their difference, the hierarchical levels, would be zero. For an undirected graph one only needs to utilise one of the forward (resp.  backward) hierarchical levels. Initial exploration shows that in the undirected scenario, where the density of the graph is greatest, the forward (resp.  backward) hierarchical levels of vertices is greatest, compartmentalising the graph into core, gateway and peripheral divisions, a functional modular partitioning. Thus shedding light on the core-periphery structure of a graph and offering a potential method of detection of said structure \cite{rombach2014core}. 

\subsection{Hierarchical differences}
\label{sec_HDs}

Hierarchical levels assign a numerical label to each vertex. We can use their differences to assign a numerical label to each edge. We will see that the weighted mean of hierarchical differences is an important metric. This leads us to the following definition.

\begin{definition}
The \textit{forward (resp.  backward) democracy coefficients} of a positively weighted simple directed graph $G$ are defined respectively to be
\begin{equation*}
\fdemocracyCoefficient{G} := 1-\mean(\fhierarchicalDifferences(G))
\end{equation*}
and
\begin{equation*}
\bdemocracyCoefficient{G} := 1-\mean(\bhierarchicalDifferences(G)),
\end{equation*}
where $\fhierarchicalDifferences(G) = \{ g_j-g_i \, |\, a_{ij} >0,\; i,j\in G \}$, the forward hierarchical differences and $\bhierarchicalDifferences(G) = \{ \gamma_i-\gamma_j \, |\, a_{ij} >0,\; i,j\in G \}$, the backward hierarchical differences. The mean is taken with respect to the edge weights.
\end{definition}

Intuitively the democracy coefficient is a measure of how much the influencers of a graph are being influenced themselves. A low democracy coefficient means that there is minimal feedback between the influencer vertices, the vertices driving the dynamics, and the influenced vertices. Similarly to the trophic incoherence parameter we define the hierarchical incoherence parameter.

\begin{definition}
The \textit{forward (resp.  backward) hierarchical incoherence parameters}, or just \textit{forward (resp.  backward) hierarchical incoherences}, of a directed graph $G$ are respectively defined to be
\begin{equation*}
\fhierarchicalincoherence{G}=\sqrt{\var(\fhierarchicalDifferences(G))}
\end{equation*}
and
\begin{equation*}
\bhierarchicalincoherence{G}=\sqrt{\var(\bhierarchicalDifferences(G))},
\end{equation*}
where the variance is computed with respect to the edge weights.
\end{definition}

The democracy coefficient in conjunction with the hierarchical incoherence for a given graph, give an insight to its topology. For a graph, high democracy coefficient and low hierarchical incoherence means that all its vertices have approximately the same hierarchical level; the graph is influenced by a large percentage of its vertices. On the other hand, low democracy coefficient and low hierarchical incoherence means that there are distinct hierarchical levels; the graph is controlled by a small percentage of its vertices. A graph is maximally hierarchical if both the democracy coefficient and the hierarchical incoherence are 0. This implies the vertices can be grouped in ``layers''. All vertices in a layer have the same hierarchical level and hierarchical levels of two layers differ by an integer. Moreover, there can only be edges from one preceding layer to the layer succeeding it. Figure \ref{fig_HL_example_1}, a directed chain, illustrates a prime example. This compares more favourably against other measures which quantify hierarchy within graphs by not being able to differentiate hierarchical structure in as precise a manner without losing information \cite{coscia2018using}.  

We can refer back to Figure \ref{fig_HL_examples} to give an intuitive feeling for how these global graph metrics characterise graph structures. In Figure \ref{fig_HL_example_1} both, the forward and backward, democracy coefficients and hierarchical incoherence parameters are 0. For Figure \ref{fig_HL_example_2} both, the forward and backward, democracy coefficients are $0.8$ and hierarchical incoherence parameters are $0.6$. In Figure \ref{fig_HL_example_3} the forward democracy coefficient is $0.75$ and forward hierarchical incoherence parameter is $0.433$. The backward democracy coefficient is 0 and backward hierarchical incoherence parameter is $2.12$. Finally, Figure \ref{fig_HL_example_4} the forward democracy coefficient is $0.729$ and forward hierarchical incoherence parameter is $0.693$. The backward democracy coefficient is $0.667$ and backward hierarchical incoherence parameter is $0.885$.

The democracy coefficient has more interesting connections with the topology of the graph. An interesting property of the democracy coefficients is that if $G$ is weakly connected, then we can calculate its forward (resp. backward) democracy coefficients by calculating only the hierarchical levels of its minimal source (resp. sink) subgraphs. This is stated and proven in the appendix Lemma \ref{lemma_mean_HDs_decomposition}

\begin{lemma}
For a weakly connected directed graph $G$, $\fdemocracyCoefficient{G}\ge0$ and $\bdemocracyCoefficient{G}\ge0$. Moreover, $G$ is a simply forward influenced graph if and only if $\fdemocracyCoefficient{G}=0$. Similarly, $G$ is a simply backward influenced graph if and only if $\bdemocracyCoefficient{G}=0$.
\label{lemma_HDs_bounded_by_1}
\end{lemma}

\begin{lemma}
If $G$ is a balanced graph, i.e. for any vertex its in-degree equals its out-degree, then $\fdemocracyCoefficient{G}=\bdemocracyCoefficient{G}=1$.
\label{lemma_degrees_equal_then_flat_graph}
\end{lemma}

By consequence of the above definition, the democracy coefficient of an undirected graph is 1. We conjecture that the democracy coefficient has also the following properties.

\begin{conjecture}
Let $G$ be a weakly connected directed graph. Then the following are true:
\begin{itemize}
\item $ \fdemocracyCoefficient{G} \le 1$ and $ \bdemocracyCoefficient{G} \le 1.$
\item $\fdemocracyCoefficient{G}=\bdemocracyCoefficient{G}=1$ if and only if the graph is balanced.
\end{itemize}
\label{conjecture_upper_bound}
\end{conjecture}

If the above upper bounds are correct, then the total weight of edges in the forward (resp.  backward) influenced subgraph bounds the forward (resp.  backward) democracy coefficient and the converse. We would get $\eta(G)\le 1-m/n$ and $m\le n(1-\eta(G))$, where $m$ is the sum of weights of all edges in the simply forward (resp.  backward) influenced subgraph and and $n$ the total weight of all edges in $G$. Moreover, we can also use the hierarchical differences to define a new graph centrality measure.

\begin{definition}
The \textit{forward (resp.  backward) influence centrality} of vertex $j$ of a positively weighted simple directed graph $G$ are respectively defined to be
\begin{equation*}
\fdemocracyCoefficientVertex{G}{j}= 1 - \mean(\fhierarchicalDifferences(G,j))
\end{equation*}
and
\begin{equation*}
\bdemocracyCoefficientVertex{G}{j}= 1 - \mean(\bhierarchicalDifferences(G,j)),
\end{equation*}
where $\fhierarchicalDifferences(G,j) = \{ g_j-g_i \, |\, a_{ij} > 0,\; i\in  G  \}$, $\bhierarchicalDifferences(G,j) = \{ \gamma_i-\gamma_j \, |\, a_{ij} > 0,\; i\in  G  \}$ and the mean was taken with respect to the edge weights with the convention that the mean of the empty set is 0.
\end{definition}

As we will see below, a vertex is an influencer if and only if its influence centrality is positive. We will also see that there is a relation between influence centrality and the stationary distribution of a random walk on the graph.

\begin{lemma}
Let $G$ be a positively weighted, weakly connected, directed graph and $i\in  G $. Then $\fdemocracyCoefficientVertex{G}{i}, \bdemocracyCoefficientVertex{G}{i}\ge 0$ and $i$ is forward (resp. backward) influenced if and only if $\fdemocracyCoefficientVertex{G}{i} = 0$ (resp. $\bdemocracyCoefficientVertex{G}{i} = 0$).

Moreover, there exist $i\in G$ such that $\fdemocracyCoefficientVertex{G}{i}>0$ and $\bdemocracyCoefficientVertex{G}{i}>0$, if and only if $G$ is strongly connected.
\label{lemma_influence_centrality_properties}
\end{lemma}

\begin{lemma}
Let $G$ be a strongly connected graph, let $\mathcal{D}$ be its weighted out-degree diagonal matrix, let $\epsilon $ be the vector of backward influence centrality of $G$
and let $\pi$ be the stationary distribution of a random walk on $G$. Then there exists $c>0$ such that $c\pi=\mathcal{D}^2\epsilon $.
\label{lemma_random_walk_connection}
\end{lemma}

In the above lemma we assumed that at each step the random walker chooses randomly an out-edge with probability proportional to that edge's weight. If the graph is weakly connected, then a random walker on it will eventually be trapped on a minimal sink subgraph. This means that we can apply Lemma \ref{lemma_random_walk_connection} on the minimal sink subgraph and get the stationary distribution of the random walk there. Moreover, it is easy to see that all vertices with 0 backward influence centrality are transient states for the random walk.

This new graph centrality measure is able to determine the vertices with non-zero influence centrality and hence those that drive the asymptotic dynamics of all vertices in the graph. The number of sources and sinks of information clearly has a strong link to the hierarchical level and influence centrality profile of a graph. In addition, the hierarchical level and influence centrality of a vertex help differentiate the vertices importance regarding control of the graph. Influencer vertices are ranked differently depending on their position in the graph and connection to the rest of network. A vertex that has more directed paths associated with it is more influential. By further studying the profile of influence centralities and the democracy coefficient one could develop a control profile/signature for graphs and use it as a comparison of controllability across different graph structures. Ruths et al \cite{ruths2014control} shows where this work has been done using the degree distribution of a graph and degree properties of a vertex. We believe influence centrality and hierarchical level could produce a more accurate representation of the control properties of a complex network.

\subsection{Related work}

There has been prior research which attempted to tackle the hierarchical nature of complex graphs as mentioned in the introduction. One major benefit of our approach compared to others see  \cite{krackhardt2014graph},\cite{corominas2013origins} is that we can account for weighted graphs and hence we can incorporate the strength of edge links into the study of hierarchical structures, particularly important for directed flow processes occurring on a graph. 

There are some approaches with a similar stratagem that have attempted to quantify hierarchy by generalising trophic analysis, see \cite{mackay2020directed} and \cite{kichikawa2019community}. There are, we believe, further improvements in our definition of trophic analysis in comparison to these methods and others \cite{luo2011detecting}. Mainly these methods rank the vertices of a strongly connected component in a directed graph at the same trophic level. In our definition, vertices of a strongly connected graph are not at the same hierarchical level, unless some symmetry is imposed on the graph. Moreover, there is a non-trivial yet explicit connection between our definition, the topology of the graph and the stationary distribution of a random walk on the graph. There is definite scope for further study and \cite{li2012digraph} shows an initial direction to explore. Additionally \cite{czegel2015random} uses the theory of random walks to provide a hierarchical analysis of complex graphs, further strengthening the connection between random walk theory and Graph Hierarchy.

Researchers have utilised the pseduoinverse of the Laplacian matrix for a graph to calculate graph metrics, such as nodal spreading capacity \cite{van2017pseudoinverse} and topological centrality \cite{ranjan2013geometry}, a structural measure, by calculating the main diagonal of the pseudoinverse. A significant computational complexity advantage is present in our approach as the underlying computational algorithm does not calculate the pseudoinverse explicitly, but rather the product of the pseudoinverse with the degree vector. This makes the algorithm more memory efficient as well as offering a significant time complexity advantage by utilising efficient convex optimisation algorithms. This also compares favourably with other approaches quantifying the hierarchical nature of graphs  which provide less information, see \cite{crofts2011googling}, \cite{mones2012hierarchy} and \cite{coscia2018using}. Additionally, graph study is typically classified into three levels:

\begin{enumerate}
    \item Microscopic - constitutes the study of individual vertices to understand their behaviour. Degree or centrality measures are typically used. Hierarchical level and influence centrality provide this level of granular analysis. 
    
    \item Mesoscopic - concerns the study of groups or community structure. At this level it is interesting to study the interaction of vertices over short distances or classification of vertices. Hierarchical differences could provide this classification via a functional clustering algorithm. Vertices within an integer hierarchical difference of each other belong to the same functional module in the graph, simultaneously addressing the order hierarchy question. 
    
    \item Macroscopic - classifies the global structure of the graph. Typically relevant parameters are average degree, average path length and average clustering coefficient. Hierarchical incoherence and democracy coefficient provide relevant parameters characterising the graph on the global scale. 
\end{enumerate}

Graph Hierarchy provides all this analysis within one coherent and comprehensive framework, a huge advantage over other approaches attempting to quantify hierarchy within complex graphs. Furthermore, within this construct we have laid the groundwork for a novel approach to controllability of complex graphs with the new metrics of influence centrality and democracy coefficient. 

Graph Hierarchy generates an intriguing relationship with normality \cite{asllani2018structure}. In the case of an undirected, the adjacency matrix is  symmetric and hence no deviation from normality, but this is not true for a directed graph. There is clearly a strong correlation between trophic coherence and the non-normality of the adjacency matrix \cite{johnson2020digraphs} and so the relationship with hierarchical incoherence requires further investigation. How these two impact the dynamics of a graph process and hence the stability requires further study. But we note the advantage of the hierarchical approach is that it has both a microscopic lens as well as a macroscopic one. Contrastingly, when analysing normality of a matrix, it's a global measure derived from the graph's adjacency matrix that is computationally inefficient by comparison.

\section{Contagion dynamics}

As an application, we look at contagion dynamics on directed graphs. We used a simple Susceptible-Infected-Susceptible epidemic model \cite{pastor2015epidemic}. Following \cite{klaise2016neurons}, we define the probability that vertex $i$ is infected at time $t+1$ to be

\begin{align*}
\mathbb P(i\text{ is infected at time }t+1) = f_i(t)^a,
\end{align*}

where $f_i(t)$ is the fraction of $i$'s in-neighbours which are infected at time $t$ and $\alpha$ is a positive parameter that controls the infection rate. The smaller $\alpha$ is, the easier it is for a vertex to be infected. We show that the forward hierarchical incoherence parameter has the same predicting ability as the trophic incoherence parameter on graphs where both can be defined, however the new generalisation can be used in more general settings.

\subsection{Graph generation}

We generate graphs using a version of the \textit{preferential preying model} (PPM), introduced in \cite{johnson2014trophic}, modified to not have ``basal'' vertices. We call this model \textit{non-source preferential preying model} (NSPPM). This modification creates graphs which are not simply forward influenced because they lack source vertices, but which are similar to PPM graphs. In Appendix \ref{graph generation} we give the algorithms for generating both PPM and NSPPM graphs. Both models have a temperature parameter, $T$, which controls how hierarchical they are. Low temperature creates graphs with low hierarchical incoherence. High temperature creates graphs which have high hierarchical incoherence, that are similar to Erd\"os-R\'enyi graphs. In the Appendix we also show that on PPM graphs the trophic incoherence parameter and the forward hierarchical incoherence parameter agree.

\begin{figure}[t]
    \centering
    \begin{subfigure}[t]{0.48\textwidth}
    {\includegraphics[height=2.6in]{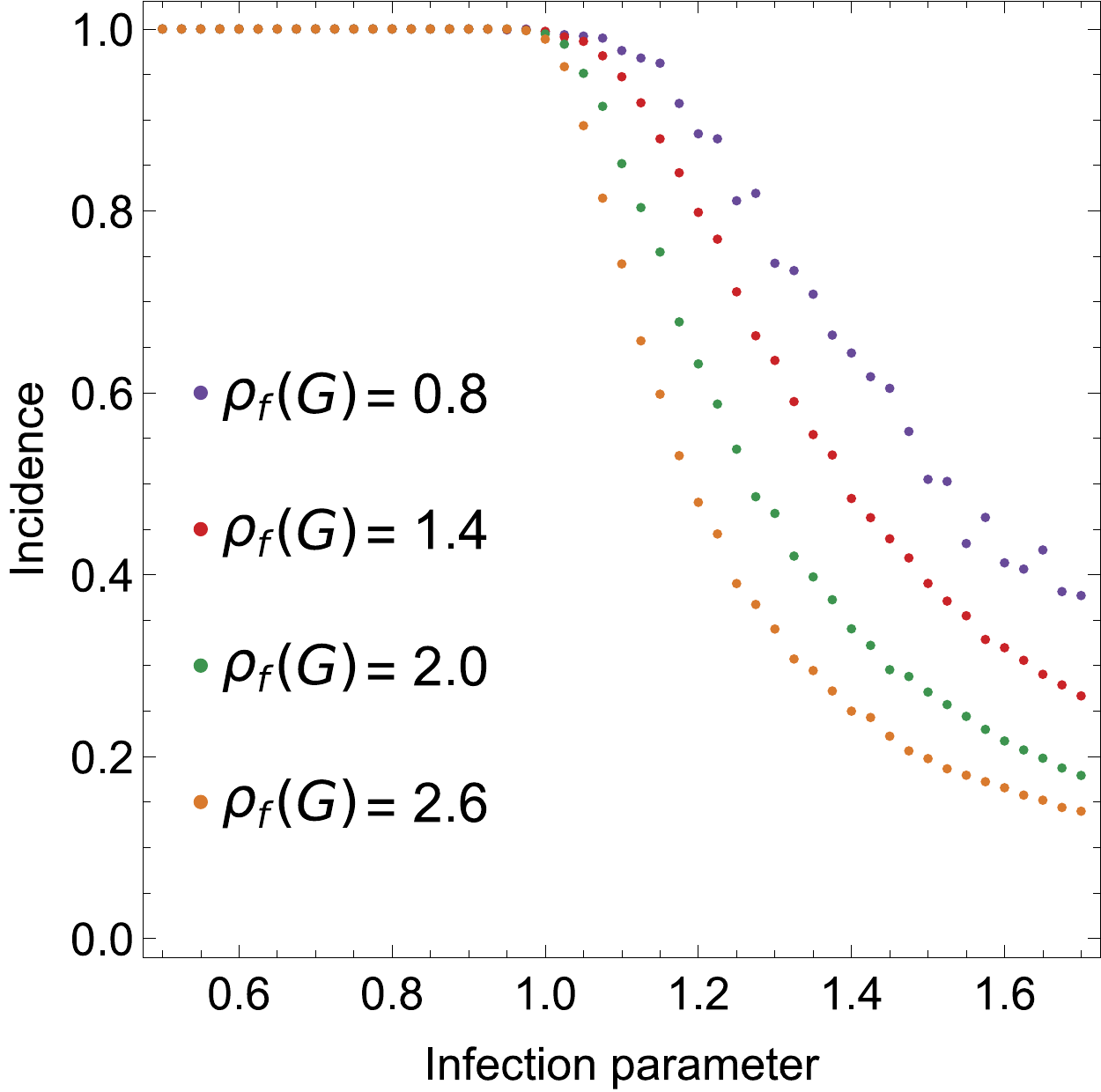}}
    \subcaption{}
    \label{fig_varying_Q_A_Incidence_mean_scatter_plot}
    \end{subfigure}~
    \begin{subfigure}[t]{0.48\textwidth}
    {\includegraphics[height=2.6in]{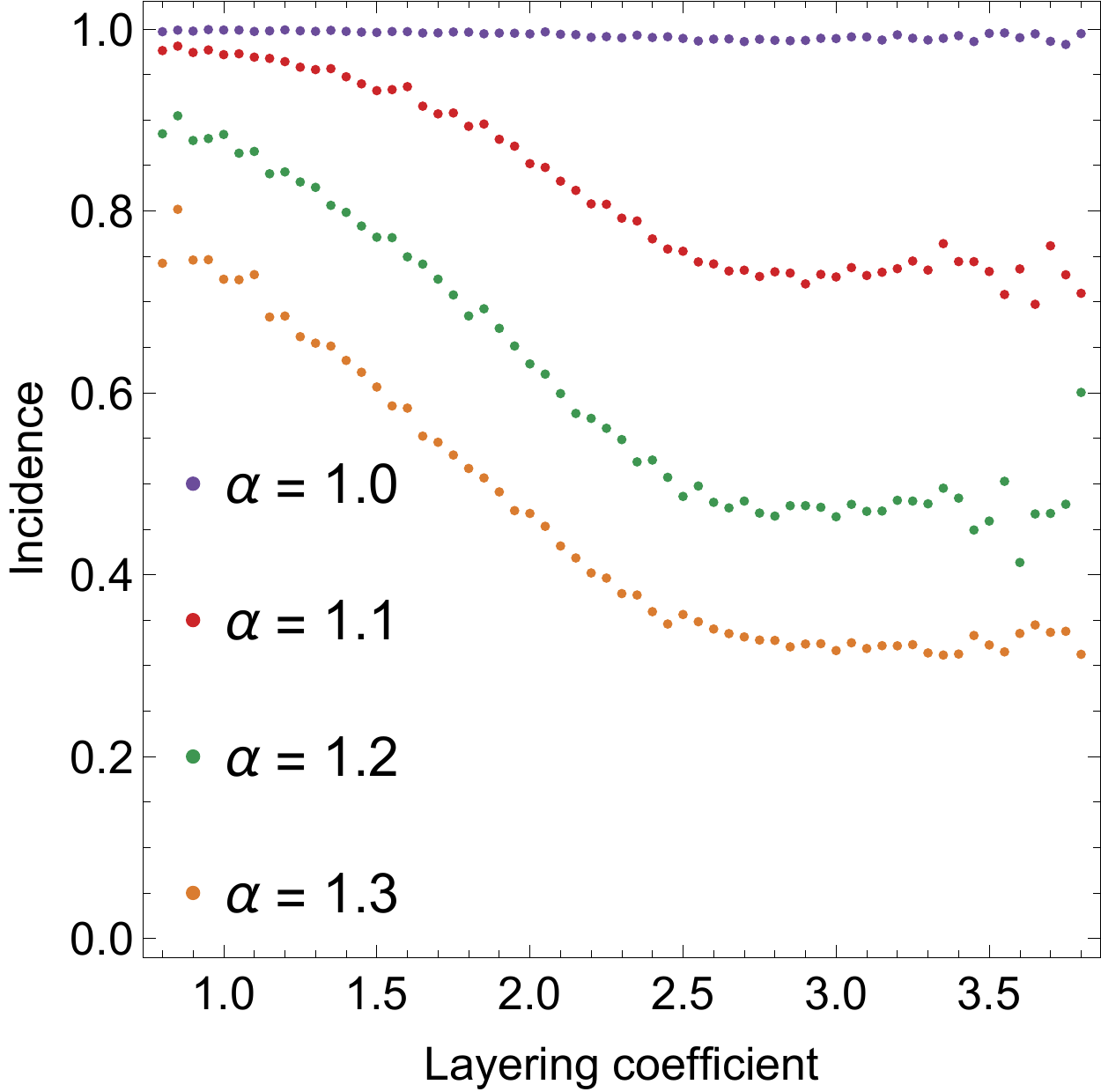}}
    \subcaption{}
    \label{fig_Q_varying_A_Incidence_mean_scatter_plot}
    \end{subfigure}
    \caption{Scatter plot of average incidence values from Monte Carlo simulations of the infection spreading with varying hierarchical incoherence $\fhierarchicalincoherence{G}$ and infection parameter $\alpha$. The average is taken over an interval of hierarchical incoherence values. (a) Incidence against $\alpha$ for different values of $\fhierarchicalincoherence{G}$. (b) Incidence against $\fhierarchicalincoherence{G}$ for different values of $\alpha$.}
    \label{fig_Q_A_Incidence_mean_scatter_plot}
\end{figure}

\subsection{Monte Carlo simulations}

We ran Monte Carlo simulations using NSPPM graphs with $500$ total vertices and $2500$ edges with varying $T$. We start the simulation by infecting the 25 vertices with the lowest hierarchical levels and measured the incidence of the infection, i.e. the proportion of vertices that have been infected at least once.
Then we run the simulation until either incidence became $1$ or there was no infected vertex left or we reached time step $1000$.

We observed that NSPPM graphs can be broken down into two categories based on the democracy coefficient. Graphs with democracy coefficient less than $20/2500$ have a few small forward influencing subgraphs which start infected and stay infected for ever. This means that depending on $\alpha$ either everything becomes infected or the algorithm times out. Graphs with democracy coefficient greater than $20/2500$ tend to have bigger forward influencing subgraphs which do not stay infected. The two types of graphs exhibit different dynamical behaviour and we are able to identify this based only on the democracy coefficient. We discuss this in more detail in Appendix \ref{graph generation} and \ref{Contagion}.

In Figure \ref{fig_Q_A_Incidence_mean_scatter_plot} we show the average incidence for graphs with democracy coefficient $20/2500$. Since we cannot choose values for the hierarchical incoherence, the average is taken over small intervals.
We see that if $\alpha \le 1$ or smaller then the average incidence is almost 1, which means that every vertex becomes infected at least once.
When $\alpha>1$, then incidence depends on the hierarchical incoherence of the graph. As expected lower hierarchical incoherence, corresponds to higher incidence.

\section{Conclusion \& Future Work}

We have generalised the previous notions of trophic levels and coherence such that a trophic analysis can be applied to any simple graph, whether it be directed or undirected. Additionally we have provided an original and fresh perspective on the nature of hierarchical organisation of complex graphs and control of them. The trophic approach has been previously applied to complex systems, particularly ecology, but with enforced restrictions that trophic analysis induces. Examples beyond the list given in the introduction, where our newly developed generalisation could advance previous work, include financial markets \cite{moran2019may} and water distribution networks \cite{pagani2020quantifying}. This now enables some key questions concerning graph dynamics to be answered in a novel way: What is the resilience and robustness of a graph and how does it adapt under failure or attack? How does the topology of a graph affects its dynamics? How is the critical component of a graph determined when a flow is spreading across the structure and how does it rewire or redistribute flow under changing conditions? 

Finally, we have developed an open source python module named \texttt{GraphHierarchy}\footnote{\href{https://github.com/shuaib7860/GraphHierarchy}{https://github.com/shuaib7860/GraphHierarchy}} and a Julia package named \texttt{GraphHierarchy}\footnote{\href{https://github.com/gmoutsin/GraphHierarchy.jl}{https://github.com/gmoutsin/GraphHierarchy.jl}} for researchers to easily apply these methods to a myriad of domains. We see this being extensively applicable as the graph abstraction is widespread in the analysis of complex systems. As such the potential of this approach to advance the study of the interplay between graph topology and dynamics is immense and leaves plenty room for future work and collaboration.

\subsection*{Contributions}
The idea of hierarchical levels was conceived by CS and developed by GM \& CS. The claims were proved by GM. The dynamics was simulated by GM. GM \& CS wrote the paper. WG \& SJ provided guidance.

\subsection*{Acknowledgements}
The authors would like to express their gratitude to Robert MacKay for his invaluable advice and showing their approach to the subject. CS also gratefully acknowledges funding from the UK Engineering and Physical Sciences Research Council under the EPSRC Centre for Doctoral Training in Urban Science (EPSRC grant no. EP/L016400/1).

\printbibliography

\newpage

\appendix
\appendixpage
\addappheadtotoc

\section{Trophic Levels and Trophic Differences}
\label{sec_trophic}

The concept of trophic levels was introduced in \cite{lindeman1942trophic} as a way to determine the hierarchy of species in a food chain. Primary producers, for example plants, have trophic level 1 and the trophic level of every other species is 1 plus the average trophic level of the species it eats. Interconnected food chains form what is called a food web. In a perfectly layered food web, all species have integer trophic levels and the difference between the trophic levels of the prey and the predator is 1. In practice this rarely happens and the notion of the trophic incoherence parameter was introduced as a way to measure how far a food web is from being perfectly layered.

\subsection{Trophic levels}

We represent a food web by a positively weighted, directed, simple graph. Typically, the direction of arrows indicate the flow of energy. We define the positive in-degree by $\tilde d_i = d_i$ if $d_i>0$ or $\tilde d_i = 1$ if $d_i=0$. We also define the positive in-degree vector $\tilde d = (\tilde d_1,\dots,\tilde d_n)$ and the positive in-degree Laplacian by $\tilde L = \tilde D - A$, where $\tilde D = \text{diag}(\tilde d)$. Finaly we define $\tilde M = \tilde L^\mathsf{T}$.
Trophic levels are defined by the following linear equations:
\begin{equation}
\begin{array}{ll}
s_i=1+\dfrac{1}{d_i}\sum_{j}a_{ji}s_j, & \quad\text{ if } d_i\ne0,\\
s_i=1 & \quad\text{ if } d_i=0.
\end{array}
\label{eq_TL_def_system}
\end{equation}

Using our notation we can write this system of equations in a compact form: $\tilde M s=\tilde d$. This leads to the following definition.

\begin{figure}[t]
    \centering
    \begin{subfigure}[t]{0.4\textwidth}
        \includegraphics[height=2.9in]{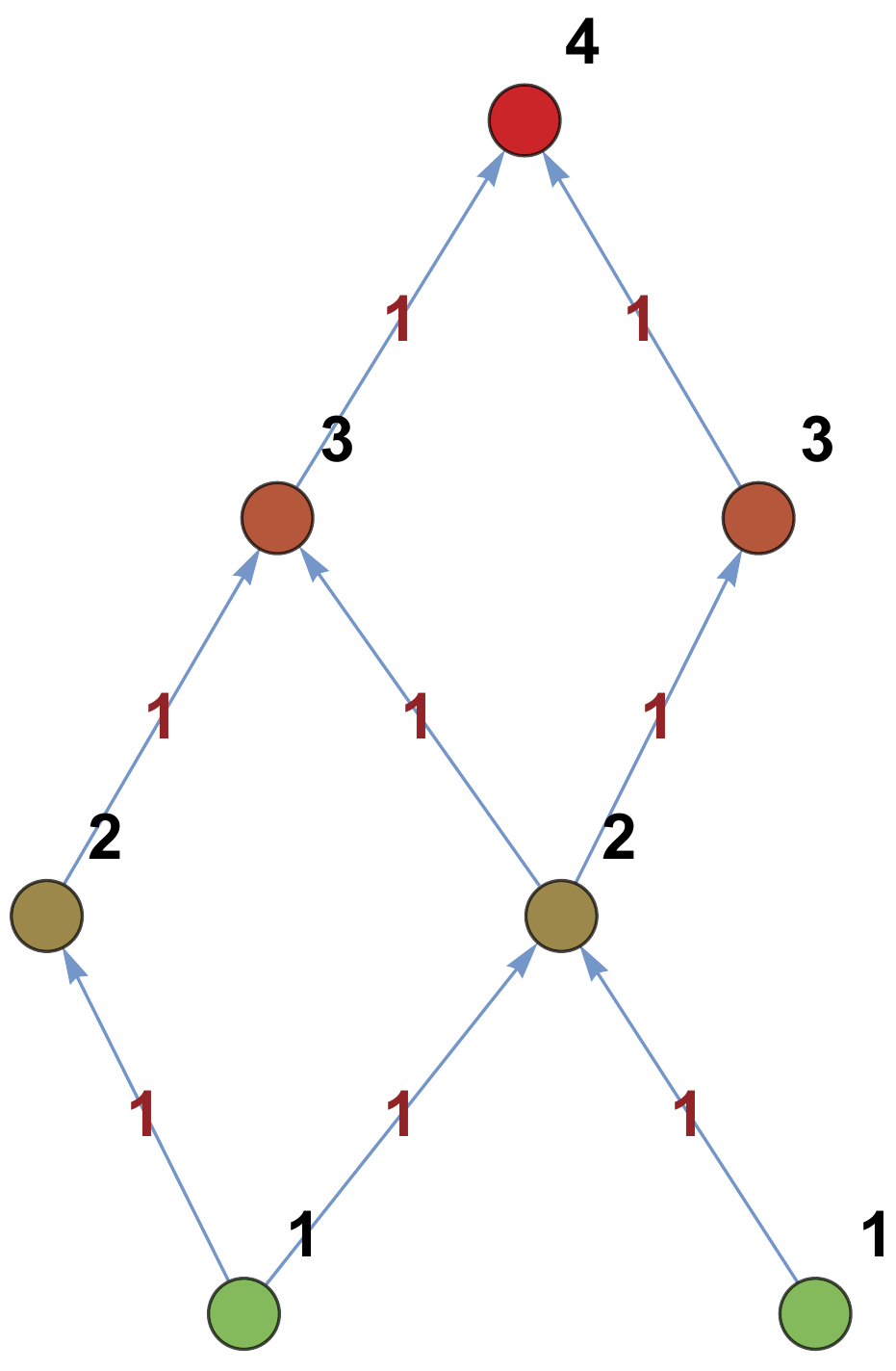}
        \subcaption{}
    \label{fig_two_graph_examples_TL_a}
    \end{subfigure}
    ~
    \begin{subfigure}[t]{0.28\textwidth}
        \includegraphics[height=2.9in]{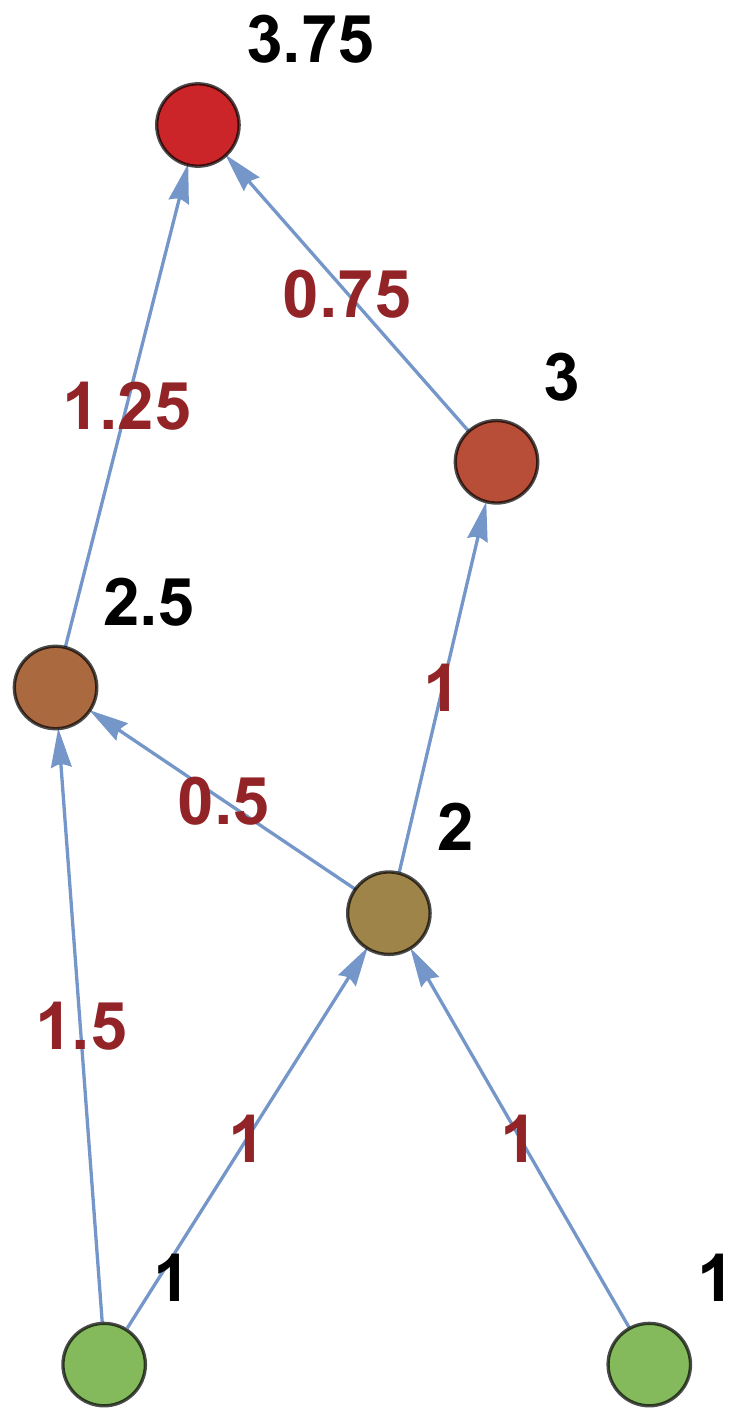}
        \subcaption{}
    \end{subfigure}
    \caption{Two graphs representing two different food webs. Trophic levels are printed in black and trophic differences in red. (a) A totally coherent graph with integer trophic levels and trophic incoherence $0$. (b) A less coherent graph with non-integer trophic levels and trophic incoherence $0.322$.}
    \label{fig_two_graph_examples_TL}
\end{figure}

\begin{definition}
Let $G$ be a simply forward influenced, positively weighted, simple graph. Then the vector of trophic levels on $G$ is $s=\tilde M^{-1} \tilde d$.
\end{definition}

The matrix $\tilde M$ is invertible if and only if the graph is simply forward influenced. An example of such graph can be seen in Figure \ref{fig_two_graph_examples_TL}.

\subsection{Trophic differences}

Trophic difference, i.e. the difference of trophic levels between two vertices connected by an edge, can be used to get a measure of how close a graph is to being perfectly layered. We define the trophic differences of a simply forward influenced graph $G$ to be the set $\trophicDifferences(G) = \{ s_j-s_i \, |\, a_{ij} > 0 \}$.

\begin{lemma}
Let $G$ be a simply forward influenced, positively weighted, simple graph, then $\mean( \trophicDifferences(G))=1$, where the mean is taken with respect to edge weights.
\label{lemma_mean_trophic_differences}
\end{lemma}

We give the proof of this lemma in Section \ref{sec_proofs}.
since the mean is always 1, the standard deviation of $\trophicDifferences(G)$ can be a measure of the distance to a perfectly layered graph and
is called the \textit{trophic incoherence parameter} or just \textit{trophic incoherence} of the graph. It is defined by
\begin{equation*}
q(G)=\sqrt{\dfrac{\sum_{ij}(s_i-s_j)^2\,a_{ij}}{\sum_{ij}a_{ij}}-1}.
\end{equation*}

\subsection{Hierarchical levels of simply forward influenced graphs}
\label{sec_HL_trophic_graphs}

In order to make the connection between trophic levels and forward hierarchical levels clear, we will discuss the case of simply forward influenced graphs more extensively. Let us consider linear system \eqref{eq_TL_def_system} and rewrite the first equation multiplied by $d_i$ and using $g $ instead of $s$ as the unknown.
\begin{equation}
d_ig_i - \sum_{j}a_{ji}g_j=d_i.
\label{eq_HL_equation}
\end{equation}
Notice that in this case, if $d_i=0$, the equation is trivially satisfied as it becomes $0=0$. Using our notation we rewrite equations \eqref{eq_HL_equation} as
\begin{equation}
M  g = d.
\label{eq_HL_matrix}
\end{equation}
Because $M$ is a singular matrix, the above linear system does not have a unique solution. For a simply forward influenced graph the dimension of the kernel of $M$ equals the number of source vertices. This means that we can get a unique solution by choosing arbitrary values of the trophic levels of the source vertices, see Lemma \ref{lemma_mean_generalized_trophic_differences_trophic_graph}. We recover the original definition of trophic levels by setting the trophic levels of all source vertices to 1. However, using this viewpoint, we see that the choice of 1 is somewhat arbitrary and any other choice is equally valid. Instead of prescribing the trophic levels of source vertices we use Definition \ref{definition_FHL_non_trophic_graphs}. In this case, because the linear system $Mx=d$ can be solved, we have $\mathcal T = \{x\in\R^n|M  x = d\}$. We find 

\begin{equation*}
g =\underset{x\in \mathcal T}\argmin\|x\|.
\end{equation*}
It is still true that $g=M^+ d$. We will not discuss backward hierarchical levels here, as there is no corresponding notion in food webs.

We see in Figure \ref{fig_two_graph_examples_HL} that the hierarchical levels of source vertices, which are typically called basal vertices in food webs, are typically not equal. This may seem strange for a food web, however it is worth noticing that the source vertex with the lowest hierarchical level is the root vertex for more food chains than the other source vertex. This shows that the hierarchical levels have the added benefit of not treating all source vertices equally and from this we can deduce which basal species are more important in a food web.

\begin{figure}[t]
    \centering
    \begin{subfigure}[t]{0.45\textwidth}
        \includegraphics[height=2.9in]{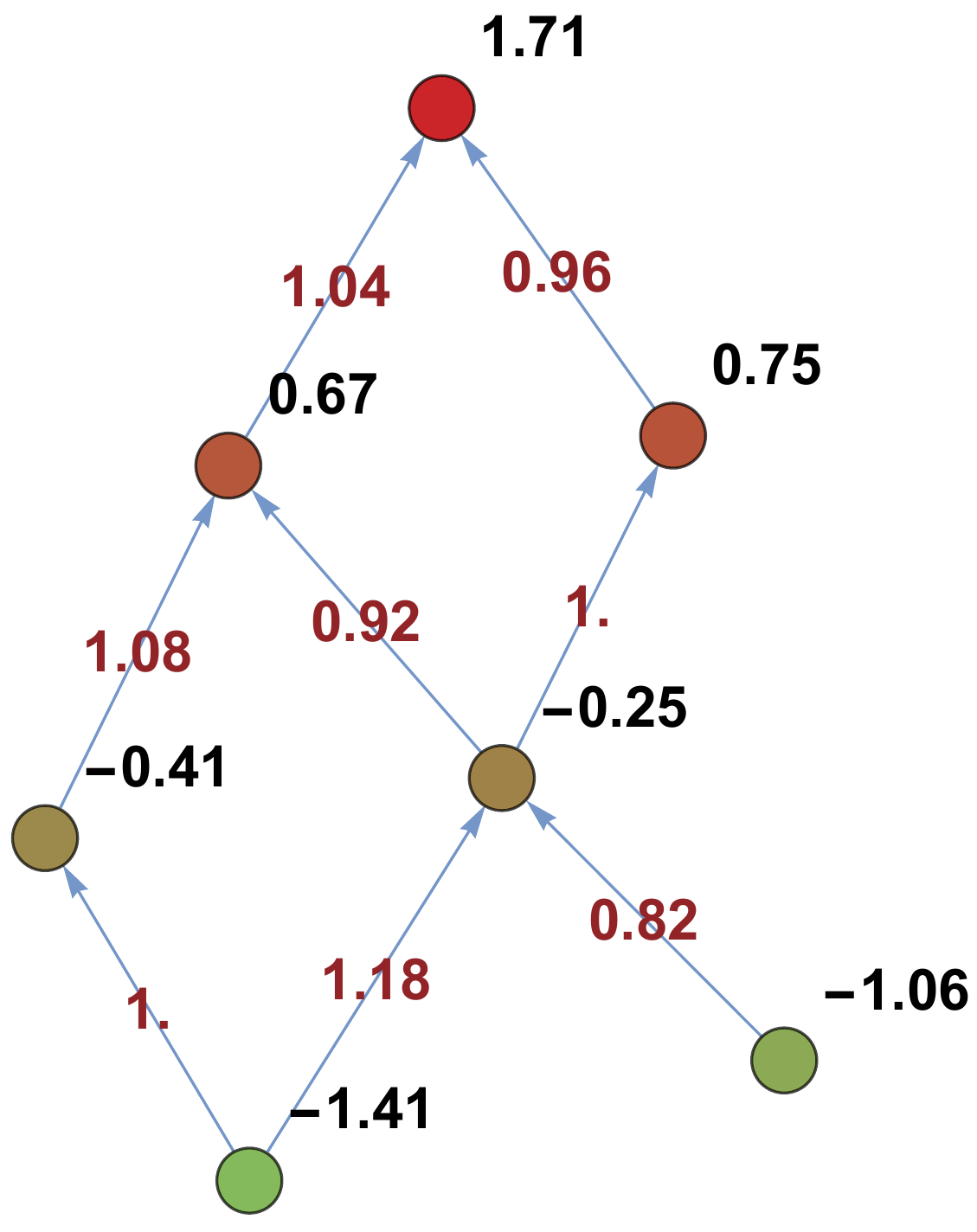}
        \subcaption{}
    \label{fig_two_graph_examples_HL_a}
    \end{subfigure}~
    \begin{subfigure}[t]{0.4\textwidth}
        \includegraphics[height=2.9in]{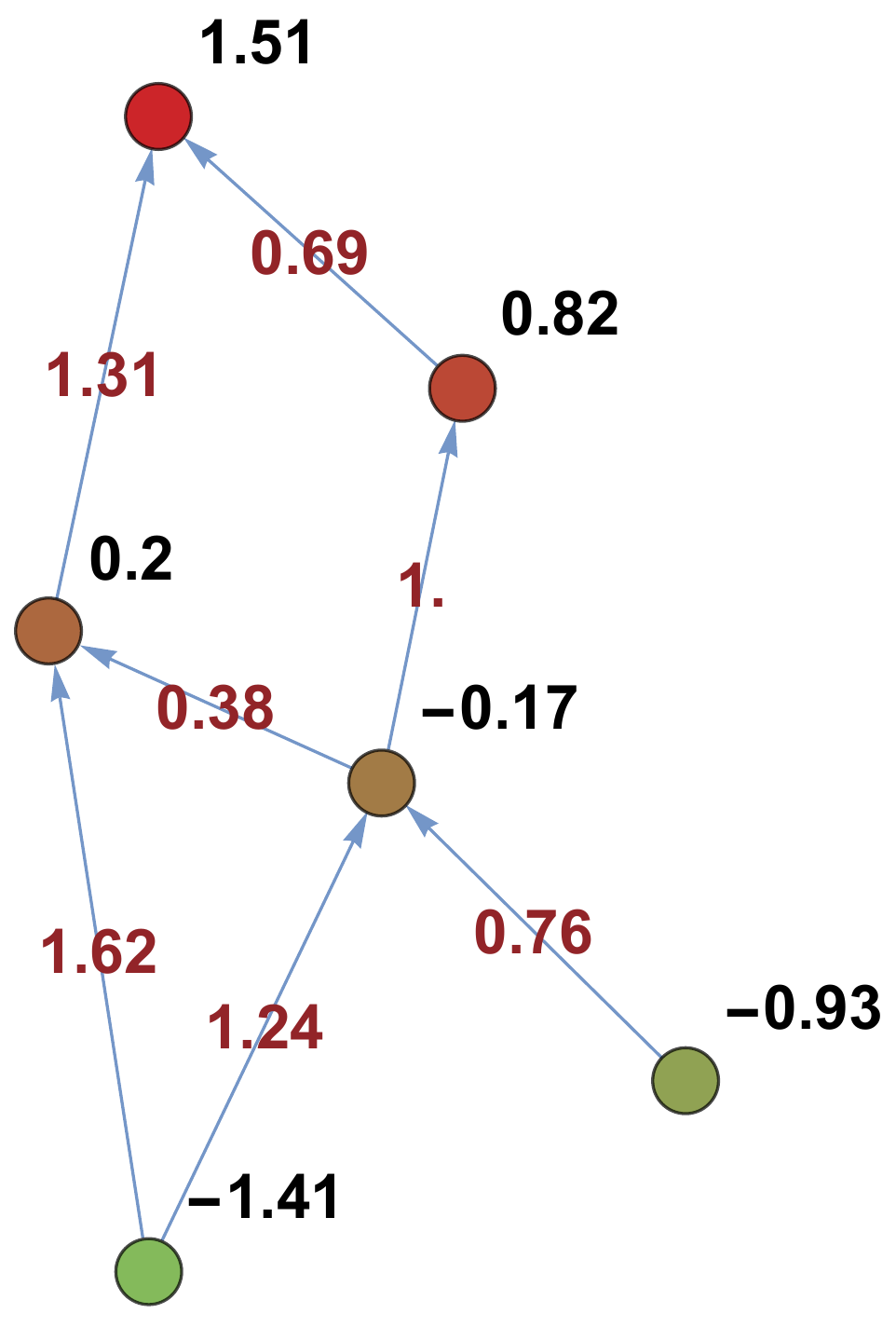}
        \subcaption{}
    \end{subfigure}
    \caption{Hierarchical levels and differences on the same graphs as in Figure \ref{fig_two_graph_examples_TL}. Forward hierarchical levels are printed in black and forward hierarchical differences in red. (a) The source vertices do not have the same forward hierarchical level and the hierarchical incoherence is $0.107$. (b) A less coherent graph with a hierarchical incoherence of $0.423$.}
    \label{fig_two_graph_examples_HL}
\end{figure}

\section{Proofs}
\label{sec_proofs}

In this section we provide the proofs of the lemmas that appear in Section \ref{sec_HDs}. The proofs are not written in the order that the lemmas appear in Section \ref{sec_HDs}, but in the order they are used in other proofs, i.e. a lemma is used in a proof only if its proof was written before.

\begin{lemma}
Let $G$ be a simply forward influenced graph with $l$ source vertices, we order the vertices of $G$ starting by the source ones. Let $d$ be its weighted in-degree vector and $L$ its weighted in-degree Laplacian. Then for any real numbers $c_1$,~$\dots$,~$c_l$ there exist real numbers $x_{l+1}$, $\dots$, $x_n$ such that the vector 
$$x=(c_1,\dots,c_l,x_{l+1},\dots,x_n)$$
satisfies
\begin{equation}
L^{T} x = d.
\label{eq_local_lemma_linear_system}
\end{equation}
Moreover, let $\mathcal D$ be the set of differences of $x$ defined by 
$$\mathcal D = \{ x_j-x_i \, |\, a_{ij} > 0,\; i,j\in  G  \}.$$
Then the weighted, by edge weights, mean of $\mathcal D$ is 1.
\label{lemma_mean_generalized_trophic_differences_trophic_graph}
\end{lemma}

\begin{proof}
Since $G$ is simply forward influenced with $l$ source vertices, we know from \cite{caughman2006kernels} that the dimension of $\ker(L)$ is $l$. We write the linear system \eqref{eq_local_lemma_linear_system} as
$$ d_i x_i - \sum_{j}a_{ji} x_j=d_i. $$
The first $l$ equations correspond to source vertices and become $0=0$, so we can choose any value for $x_i$, $i\in\{1,\dots,l\}$. Moreover, since the dimension of $\ker(L)$ is $l$, the rest of the equations can be solved. So we conclude that such $x$ exists.

We have
\begin{align*}
\mean(\mathcal D) &= \dfrac{\sum_i\sum_ja_{ji}(x_i-x_j)}{\sum_i\sum_ja_{ji}} \\
&=\dfrac{\sum_i(\sum_ja_{ji}x_i-\sum_ja_{ji}x_j)}{\sum_id_i} \\
&=\dfrac{\sum_i(d_i x_i-\sum_ja_{ji}x_j)}{\sum_id_i} \\
&=\dfrac{\sum_id_i}{\sum_id_i} =1. \qedhere
\end{align*}
\end{proof}

Lemma \ref{lemma_mean_trophic_differences} is a straightforward corollary.

\begin{lemma}
Let $G$ be a hierarchically decomposable graph, let $\Gamma_1$,~$\dots$,~$\Gamma_l$ be its minimal source (resp. sink) subgraphs and let $H$ be the simply forward (resp.  backward) influenced subgraph of $G$. Let $m$ be the sum of weights of all edges in $H$ and $l_i$ the sum of weights of all edges in $\Gamma_i$. Then
\begin{equation*}
\democracyCoefficient{G}=\dfrac{\sum_i \democracyCoefficient{\Gamma_i} l_i}{m+\sum_i l_i},
\end{equation*}
where $\eta$ denotes the forward or backward democracy coefficient respectively.
\label{lemma_mean_HDs_decomposition}
\end{lemma}

\begin{proof}[Proof of Lemma \ref{lemma_mean_HDs_decomposition}]
Let $\chi_G$ be the weighted sum of hierarchical differences of graph $G$, weighted by the edge weights.
Trivially it is true that
\begin{equation*}
\chi_G=\chi_H+\sum_i \chi_{\Gamma_i}.
\end{equation*}
From Lemma \ref{lemma_mean_generalized_trophic_differences_trophic_graph} we know that no matter what are the values of the source vertices of $H$, the weighted mean of differences will be $1$. This means that $\chi_H = m$. 
Since the weighted mean of forward differences of $\Gamma_i$ is $1-\fdemocracyCoefficient{\Gamma_i}$, we have $\chi_{\Gamma_i} = (1-\fdemocracyCoefficient{\Gamma_i}) l_i$. From this we get 
\begin{align*}
\chi_G = \chi_h + \sum_i \chi_{\Gamma_i} = m + \sum_i (1-\fdemocracyCoefficient{\Gamma_i}) l_i.
\end{align*}
Then we have
\begin{align*}
\fdemocracyCoefficient{G} &= 1 - \dfrac{\chi_G}{m + \sum_i l_i} \\
&= \dfrac{m + \sum_i l_i-\chi_G}{m + \sum_i l_i} \\
&= \dfrac{m + \sum_i l_i-m - \sum_i (1-\fdemocracyCoefficient{\Gamma_i}) l_i}{m + \sum_i l_i} \\
&= \dfrac{\sum_i \fdemocracyCoefficient{\Gamma_i} l_i}{m + \sum_i l_i}. 
\end{align*}
This proves the lemma for the forward democracy coeffient.
By doing the same for $G^T$ we prove the lemma for the backward democracy coeffient.
\end{proof}

Based on we Lemma \ref{lemma_mean_HDs_decomposition}, we conjecture that the democracy coefficient of a graph cannot be arbitrarily small.

\begin{conjecture}
Let $G$ be a directed (unweighted) graph with $m$ edges. Then
\begin{equation*}
\fdemocracyCoefficient{G}, \bdemocracyCoefficient{G}\not\in(0,\tfrac{2}{m})\cup(\tfrac{2}{m},\tfrac{3}{m}).
\end{equation*}
Moreover, if $m>3$, then $\fdemocracyCoefficient{G}=2/m$ if and only if $G$ is weakly connected and its minimal source subgraphs are all source vertices except one which is an source pair, i.e. a strongly connected subgraph with 2 vertices.
\end{conjecture}

\begin{lemma}
A minimal source (or sink) subgraph is strongly connected.
\label{lemma_min_s_subgraph_strongly_connected}
\end{lemma}

\begin{proof}
Let $G$ be a source subgraph that is not strongly connected. This means that there exist vertices $i$ and $j$ such that there is no directed path from $i$ to $j$. We define $J$ to be the set of all vertices from which there is a directed path to $j$. We define $J^c$ to be the set of all vertices from which there is no directed path to $j$. By definition $i\in J^c$. By construction there are no directed edges from $J^c$ to $J$, which implies that $J^c$ is an source subgraph of $G$, which is a contradiction.

Since a minimal sink subgraph of $G$ is a minimal source subgraph of $G^c$, we get that minimal sink subgraphs are also strongly connected.
\end{proof}

\begin{lemma}
A weakly connected graph is hierarchically decomposable if and only if it is not strongly connected.
\label{lemma_decomposable_iff_not_strongly_connected}
\end{lemma}

\begin{proof}
If the graph is strongly connected, then the only source and sink subgraph is the graph itself, so by definition it is not hierarchically decomposable.

For the other direction, we assume that the graph is not strongly connected and we repeat the construction of $J$ and $J^c$ of the proof of Lemma \ref{lemma_min_s_subgraph_strongly_connected}.
From $J^c$ we can construct $\Gamma_1$, a minimal source subgraph of $G$ and since we know that $\Gamma_1$ is not the whole $G$, we get that $G$ is hierarchically decomposable.
\end{proof}

\begin{corollary}
Let $G$ be a weakly connected graph. $G$ is hierarchically decomposable, if and only if $G^T$ is hierarchically decomposable.
\end{corollary}

\begin{proof}
This is a direct corollary of Lemma \ref{lemma_decomposable_iff_not_strongly_connected}, since $G$ is strongly connected if and only if $G^T$ is strongly connected.
\end{proof}

\begin{lemma}
Let $G$ be a strongly connected graph. Then $\ker(L)$ is spanned by a positive integer vector.
\label{lemma_strongly_connected_implies_positive_vector}
\end{lemma}

\begin{proof}
As $G$ is strongly connected, the kernel of $L$ is 1-dimensional, see \cite{caughman2006kernels}. Moreover, Proposition 4.1 in \cite{bjorner1992chip} shows that there exists a positive integer vector that belongs to $\ker(L)$. These two facts prove the lemma.
\end{proof}

\begin{lemma}
Let $G$ be a hierarchically decomposable directed graph and let $\Gamma_1$,~$\dots$,~$\Gamma_l$ be its minimal source subgraphs. Let $d$ be its in-degree vector, $L$ be its in-degree Laplacian and $L_i$ be the in-degree Laplacian of $\Gamma_i$. Then
\begin{enumerate}
\item $\ker(L_i)$ is spanned by a positive vector $\kappa_i$.
\item $\ker(L)$ is spanned by the non-negative vectors $k_i=(0,\dots,0,\kappa_i,0,\dots,0)$, where $i\in\{i,\dots,l\}$ and the position of $\kappa_i$ in $k_i$ corresponds to the position of $L_i$ in $L$.
\item $k_i d=0$ if $\Gamma_i$ is just a single vertex and $k_i d>0$ otherwise.
\end{enumerate}
\label{lemma_kernel_L_construction}
\end{lemma}

\begin{proof} \hspace{0em}
\begin{enumerate}
\item Since $\Gamma_i$ is a minimal source subgraph, by Lemma \ref{lemma_min_s_subgraph_strongly_connected} it is strongly connected and by Lemma \ref{lemma_strongly_connected_implies_positive_vector}, $\ker(L_i)$ is spanned by a positive vector $\kappa_i$.
\item Since there are $l$ minimal source subgraphs, the dimension of $\ker(L)$ is $l$, see \cite{caughman2006kernels}. By renaming the vertices, the Laplacian $L$ can be brought to the form
\begin{align*}
L=\begin{pmatrix}
L_1   & 0     & \dots & 0     & C_1\\
0     & L_2   & \dots & 0     & C_2\\
\vdots& \vdots& \ddots& \vdots& \vdots\\
0     & 0     & \dots & L_l   & C_l\\
0     & 0     & \dots & 0     & C_{l+1}
\end{pmatrix}.
\end{align*}
It is straightforward to check that the vector $k_i=(0,\dots,0,\kappa_i,0,\dots,0)$ is in $\ker(L)$. Since we can construct $l$ such vectors and by construction they are orthogonal, they forming a basis of $\ker(L)$.
\item Without loss of generality we we will consider $\Gamma_1$. If $\Gamma_1$ is a single vertex then $L_1$ is just the $1\times1$ zero matrix. This means that $k_1=(1,0,\dots,0)$ and $d=(0,d_2,\dots,d_n)$, thus $k_1 d=0$. If $\Gamma_1$ is a strongly connected graph with $m$ vertices, then $\kappa_1$ is a positive $m$-vector and the in-degree vector has the form $d=(d_1,\dots,d_m,\dots,d_n)$. This means that $d_1,\dots,d_m>0$, so we get that $k_1 d>0$. \qedhere 
\end{enumerate}
\end{proof}

\begin{lemma}
Let $G$ be a weakly connected graph, $L$ be its in-degree Laplacian and $d$ be its in-degree vector. Then a vector $x$ that satisfies $L^{T} x =d$ exists if and only if $G$ is simply forward influenced.
\label{lemma_trophic_levels_iff_trophic_graph}
\end{lemma}

\begin{proof}
Lemma \ref{lemma_mean_generalized_trophic_differences_trophic_graph} states that if a graph is simply forward influenced, then the system can be solved. 

For the converse we recall from linear algebra that such $x$ exists if and only if the orthogonal projection of $d$ onto $\ker(L)$ is the $0$ vector. We assume that there exists a vector $x$ that satsfies $L^{T} x =d$.

Let $\Gamma_1$,~$\dots$,~$\Gamma_l$ be the minimal source subgraphs $G$.
Let $k_i$, where $i\in\{1,\dots,l\}$, be the vectors that span $\ker(L)$.
Since the vector $x$ exists, this means that $d k_i=0$ for all $i\in\{1,\dots,l\}$. Then by virtue of Lemma \ref{lemma_kernel_L_construction}, for all $i$ $\Gamma_i$ is a single vertex, thus $G$ is simply forward influenced.
\end{proof}

\begin{lemma}
Let $G$ be a simple directed graph, let $d$ be its in-degree vector and let $L$ be its in-degree Laplacian. Let $b$ be the orthogonal projection of $d$ onto $\ker(L)$. Then $b$ is a non-negative vector and
\begin{equation*}
\fdemocracyCoefficient{G}=\dfrac{\sum_i b_i}{\sum_i d_i}.
\end{equation*}
\label{lemma_relation_between_mean_and_projection_onto_kerL}
\end{lemma}

\begin{proof}
Let $g$ be the vector of forward hierarchical levels of $G$. From the definition of Section \ref{sec_HDs} we have
\begin{align*}
\fdemocracyCoefficient{G} &=1- \dfrac{\sum_i\sum_ja_{ji}( g_i- g_j)}{\sum_i\sum_ja_{ji}}.
\end{align*}
We define 
\begin{equation*}
b_i := d_i - \left(d_i g_i - \sum_j a_{ji}  g_j \right)
\end{equation*}
and the vector $b := (b_1,\dots,b_n)$. This means that
\begin{align*}
b = d-M g = d-M M^+ d = (I-M M^+)d.
\end{align*}
The matrix $I-M M^+$ is the orthogonal projector onto the kernel of $M^\mathsf{T}=L$, see \cite{golub1996matrix}. So $b$ is indeed the orthogonal projection of $d$ onto $\ker(L)$. Lemma \ref{lemma_kernel_L_construction} shows that the kernel of $L$ is spanned by non-negative vectors. Since $d$ is also a non-negative vector, the projection of $d$ onto $\ker(L)$ is a non-negative vector, so $\sum_ib_i\ge0$.

We have
\begin{align*}
\fdemocracyCoefficient{G} &= \dfrac{\sum_i\sum_ja_{ji}-\sum_i\sum_ja_{ji}( g_i- g_j)}{\sum_i\sum_ja_{ji}} \\
&=\dfrac{\sum_i(d_i-d_i  g_i+\sum_ja_{ji} g_j)}{\sum_id_i} \\
&=\dfrac{\sum_ib_i}{\sum_id_i}. \qedhere
\end{align*}
\end{proof}

\begin{proof}[Proof of Lemma \ref{lemma_HDs_bounded_by_1}]
Since $\fdemocracyCoefficient{G}=\bdemocracyCoefficient{G^T}$, we will prove the forward version of the Lemma only. The backward version of the Lemma is proved by repeating the process for $G^T$.

We know from
Lemma \ref{lemma_relation_between_mean_and_projection_onto_kerL} that $\fdemocracyCoefficient{G}=\sum_i b_i/\sum_i d_i$ and $\sum_i b_i>0$. This proves the first assertion of the lemma.

The second assertion will be proved in two steps. Let $ g$ be the vector of HLs of $G$. First assume that $G$ is simply forward influenced. This means that the forward hierarchical levels vector $g$ satisfies the equation $M g=d$, i.e. $d_i  g_i-\sum_ja_{ji} g_j = d_i$ for all $i$. This gives
\begin{align*}
\fdemocracyCoefficient{G} &=1- \dfrac{\sum_i\sum_ja_{ji}( g_i- g_j)}{\sum_i\sum_ja_{ji}} \\
&=1-\dfrac{\sum_i(d_i  g_i-\sum_ja_{ji} g_j)}{\sum_id_i} \\
&=1-\dfrac{\sum_id_i}{\sum_id_i}=0.
\end{align*}
Now we assume that $G$ is a weakly connected graph with $ \fdemocracyCoefficient{G}=0 $, thus $\sum_ib_i=0$. Since $b$ is a non-negative vector, this implies that $b=0$. This implies that the projection of $d$ onto the kernel of $M^\mathsf{T}$ is $0$ and that $d$ is in the range of $M$. From this we deduce that the linear system $M g=d$ can be solved and we use Lemma \ref{lemma_trophic_levels_iff_trophic_graph} to deduce that $G$ is simply forward influenced.
\end{proof}

\begin{proof}[Proof of Lemma \ref{lemma_degrees_equal_then_flat_graph}]
First we prove that a weakly connected, balanced graph is strongly connected. We assume that the graph is not strongly connected and we separate $G$ into a source subgraph $\Gamma$ and its complement $G\setminus\Gamma$. We know that there cannot be a directed edge from $G\setminus\Gamma$ to $\Gamma$, but there has to be at least one directed edge from $\Gamma$ to $G\setminus\Gamma$. However, since the sum of in-degrees in $\Gamma$ equals the sum of out-degrees, this is impossible, so $G$ is strongly connected.

Let $L$ be the in-degree Laplacian of $G$. Because $G$ is balanced, every row and every column of $L$ sums to $0$. From this we deduce that the vector $\mathbb 1=(1,\dots,1)$ is in the kernel of both $L$ and $L^\mathsf{T}$. Since $G$ is strongly connected, the kernel of $L$ is $1$-dimensional, see \cite{caughman2006kernels}. So $\mathbb 1$ spans both $\ker(L)$ and $\ker(L^\mathsf{T})$.

The projection of $d$ onto $\ker(L)$ is
\begin{align*}
b&=\frac{d\cdot \mathbb 1}{\mathbb 1\cdot\mathbb 1}\mathbb 1
=\dfrac{\sum_id_i}{n}\mathbb 1.
\end{align*}
This means that $\sum_ib_i=\sum_id_i$. Then by Lemma \ref{lemma_relation_between_mean_and_projection_onto_kerL} we get $\fdemocracyCoefficient{g}=1$. Since $G^T$ is also balanced, we have $\bdemocracyCoefficient{G} = \fdemocracyCoefficient{G^T}=1$.
\end{proof}

\begin{proof}[Proof of Lemma \ref{lemma_influence_centrality_properties}]
For the forward version of the Lemma, we know that if $d_i=0$ then the vertex $i$ is a source subgraph and by definition $\fdemocracyCoefficientVertex{G}{i}=1$.

Assume that $d_i>0$.
We use the vector $b$ of Lemma \ref{lemma_relation_between_mean_and_projection_onto_kerL}. We have
\begin{align*}
\fdemocracyCoefficientVertex{G}{i} &=1- \dfrac{\sum_j a_{ji}( g_i- g_j)}{\sum_j a_{ji}}\\
&= 1- \dfrac{\sum_j a_{ji} g_i-\sum_j a_{ji} g_j}{d_i}\\
&= 1- \dfrac{d_i g_i-\sum_j a_{ji} g_j}{d_i}\\
&= 1- \dfrac{d_i - b_i}{d_i}\\
&= \dfrac{b_i}{d_i}.
\end{align*}
Recall that $b$ is the orthogonal projection of $d$ on $\ker(L)$. We use Lemma \ref{lemma_kernel_L_construction} and we see that for any $i$ with $d_i>0$, $b_i=0$ if and only if $i\in G\setminus\cup_i \Gamma_i$. The same argument for $G^T$ proves the backward version of the Lemma.

For the last part of the lemma, we know from Lemma \ref{lemma_strongly_connected_implies_positive_vector} that if the graph is strongly connected, then the $\ker(L)$ is spanned by a positive vector, so for every $i$, $b_i>0$.

For the converse, we assume that there exist $i\in G$ such that $\fdemocracyCoefficientVertex{G}{i}>0$ and $\bdemocracyCoefficientVertex{G}{i}>0$. Since  $\fdemocracyCoefficientVertex{G}{i}>0$, then there exists a minimal source subgraph $\Gamma$, such that $i\in \Gamma$. Similarly, since $\bdemocracyCoefficientVertex{G}{i}>0$, then there exists a minimal sink subgraph $\Delta$, such that $i\in \Delta$.
From Lemma \ref{lemma_min_s_subgraph_strongly_connected} we know that both $\Gamma$ and $\Delta$ are strongly connected.
Moreover, since they have a common vertex, they are actually the same subgraph. This means that since $\Gamma$ is both a minimal source and minimal sink subgraph, there are no directed edges between $\Gamma$ and the rest of $G$. This means that since $G$ is weakly connected $\Gamma = G$.
This concludes the proof.
\end{proof}

\begin{proof}[Proof of Lemma \ref{lemma_random_walk_connection}]
Let $D$ be $G$'s weighted in-degree diagonal matrix
Since $G$ is strongly connected, the vectors $d$ and $\delta$ have no 0 entries, which implies that the matrices $D$ and $\mathcal{D}$ are invertible. In the proof of Lemma \ref{lemma_influence_centrality_properties}, we saw that 
$ \fdemocracyCoefficientVertex{G}{i} = b_i/d_i $
with $b$ projection of $d$ onto $\ker(M^T)$. This relation can be written as $e = D^{-1}b$, where $e$ the vector of forward influence centrality of $G$.

The same argument for $G^T$ gives us $\epsilon = \mathcal{D}^{-1}\beta$, where $\epsilon $ is the vector of backward influence centrality and $\beta$ is the projection of the weighted out-degree vector $\delta$ onto $\ker(\Lambda^T)$.
We assume that the random walker on $G$ follows an edge with probability proportional to that edge's weight.
Let be $P$ the transition probability matrix of $G$, then $A=\mathcal{D} P$, where $A$ the adjacency matrix of $G$.

We define $Q=I-P^T$ and from the fundamental theorem of Markov chains we know that the kernel of $Q$ is one-dimensional and that $Q\pi=0$.
It holds that $Q \mathcal{D} = (I-P^T) \mathcal{D} = \mathcal{D} - P^T \mathcal{D} = \mathcal{D} - A^T = \Lambda^T$.

Since $\beta\in \ker(\Lambda^T)$ and $ \mathcal{D} \epsilon=\beta$, then $ \mathcal{D} \epsilon\in \ker(\Lambda^T)$. Which implies that $\Lambda^T \mathcal{D} \epsilon=0$ and since $Q \mathcal{D} =\Lambda^T$, we get $Q \mathcal{D} ^2\epsilon=0$. Because the kernel of $Q$ is one-dimensional, $\pi$ and $ \mathcal{D}^2\epsilon$ are parallel vectors. This concludes the proof.
\end{proof}

\section{Graph Generation}
\label{graph generation}

\subsection{Preferential Preying Model}

The \textit{Preferential Preying Model} (PPM) was introduced in \cite{johnson2014trophic} as a way to generate graphs that are similar to food webs.
In order to generate a graph with PPM we choose the number of vertices $N$, the number of source vertices $B$, the number of edges $E$ and the ``temperature'' $T$.
The PPM algorithm is:
\begin{enumerate}
\item We introduce $B$ source vertices and no edges.
\item We choose uniformly at random one of the existing vertices $i$ and we add a new vertex $j$ and the edge $i\to j$.
\item We repeat step $2$ until we have $N$ vertices in total.
\item We assign each vertex $i$ its trophic level $s_i$ according to the graph we have up to this point.
\item From all possible edges $i\to j$ such that $j$ is not an source vertex, we choose $L-N+B$ with probability proportional to
\begin{equation*}
\mathbb P(a_{ij}=1) \propto \text{exp}\left(-\frac{(s_j-s_i-1)^2}{2T^2}\right).
\end{equation*}
\end{enumerate} 

All PPM graphs are simply influenced graphs, so the democracy coefficient is $0$. In Figure \ref{fig_PPM_TL_HL_comparison} we see that the trophic incoherence coefficient almost equals the forward hierarchical incoherence coefficient if they are not small. This is due to the fact that a perfectly layered graph will have 0 trophic incoherence because its source vertices have all the same level. In contrast the forward hierarchical levels of the source vertices vary depending on the connectivity of the graph, so the forward hierarchical incoherence is rarely 0, nicely displayed by comparing the graphs in Figures \ref{fig_two_graph_examples_TL_a} and \ref{fig_two_graph_examples_HL_a}.

\begin{figure}[t]
    \captionsetup{format=plain}
    \centering
    \begin{subfigure}[t]{0.31\textwidth}
    \includegraphics[height=1.6in]{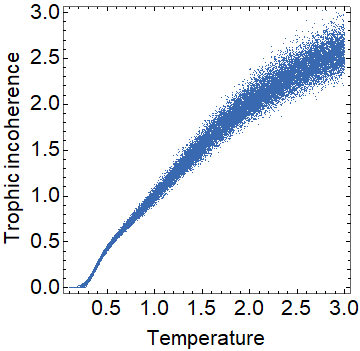}
    \subcaption{}
    \end{subfigure}~
    \begin{subfigure}[t]{0.31\textwidth}
    \includegraphics[height=1.6in]{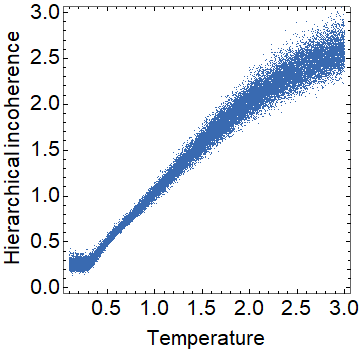}
    \subcaption{}
    \end{subfigure}~
    \begin{subfigure}[t]{0.31\textwidth}
    \includegraphics[height=1.6in]{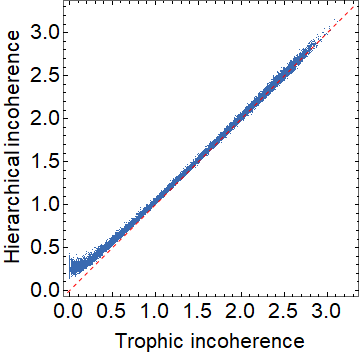}
    \subcaption{}
    \end{subfigure}
    \caption{The correlation between temperature, trophic incoherence and hierarchical incoherence in PPM graphs. (a) Scatter plot of trophic incoherence over temperature. (b) Scatter plot of forward hierarchical incoherence over temperature. (c) Scatter plot of forward hierarchical incoherence over trophic incoherence. Notice that there is some divergence between them only for small values of trophic incoherence.}
    \label{fig_PPM_TL_HL_comparison}
\end{figure}

\subsection{Non-Source preferential preying model}
\label{sec_NSPPM}

\begin{figure}[t]
    \captionsetup{format=plain}
    \centering
    \begin{subfigure}[t]{0.48\textwidth}
    \includegraphics[height=2.1in]{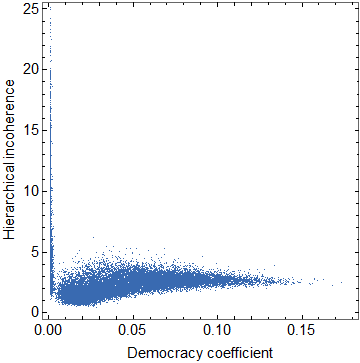}
    \subcaption{}
    \end{subfigure}~~~
    \begin{subfigure}[t]{0.48\textwidth}
    \includegraphics[height=2.1in]{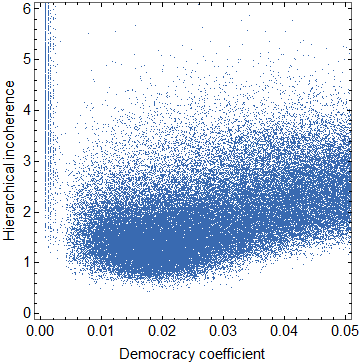}
    \subcaption{}
    \end{subfigure}
    \caption{The scatter plots of hierarchical incoherence over democracy coefficient for NSPPM graphs. Two different regions are visible, one with democracy coefficient between 0 and $20/2500$ and hierarchical incoherence that go up to 25 and another band with hierarchical incoherence greater than $20/2500$ and hierarchical incoherence between $0.6$ and 5. (a)A scatter plot with hierarchical incoherence between $0$ and $60$. (b) A zoomed in scatter plot with hierarchical incoherence between $0$ and $6$.}
    \label{fig_GTC_Mean_scatter_plot}
\end{figure}

The \textit{non-source preferential preying model} (NSPPM) algorithm is a modification of the PPM algorithm. In practice we generate a PPM graph and then we make sure that there are no source vertices. We choose the number of vertices $N$, the number of source vertices $B$, the number of edges $E$ and the ``temperature'' $T$. The NSPPM algorithm is:
\begin{enumerate}
\item We introduce $B$ source vertices and no edges.
\item We choose uniformly at random one of the existing vertices $i$ and we add a new vertex $j$ and the edge $i\to j$.
\item We repeat step $2$ until we have $N$ vertices in total.
\item We assign each vertex $i$ its trophic level $s_i$ according to the graph we have up to this point.
\item From all possible edges $i\to j$ such that $j$ is not an source vertex, we choose $L-N+B$ with probability proportional to
\begin{equation*}
\mathbb P(a_{ij}=1) \propto \text{exp}\left(-\frac{(s_j-s_i-1)^2}{2T^2}\right).
\end{equation*}
\item We pick an source vertex $i$ with in-degree $0$, we pick another vertex $j$ with probability proportional to $\text{exp}(-s_j)$ and we add the edge $j\to i$.
\item We repeat step $6$ until all source vertices have in-degree $1$.
\end{enumerate}

We find that NSPPM graphs can be separated into two types. The ones with very small democracy coefficient (smaller than $20/2500$) and the one with a democracy coefficient bigger than $20/2500$. We can see in Figures \ref{fig_GTC_Mean_scatter_plot} and \ref{fig_NSPPM_scatter_plots} the two different types produce very different distributions of hierarchical incoherence. The value $20/2500$ was chosen empirically based on the results. If Conjecture \ref{conjecture_upper_bound} is true, then it would mean that for the graphs in the first category the sum of edges in their source subgraphs is 20 or lower, which we always found to be the case. We found that roughly $2\%$ of the generated graphs fell into this category. This percentage varied a bit with $T$, with lower $T$ having higher probability of generating this type of graph. We see that graphs of the first type have a very wide range of hierarchical incoherence. This depends on the connectivity of the source subgraph and the topology of the simply forward influenced subgraph. When the imply forward influenced subgraph has no clear/strong hierarchy, i.e. resembles a Erd\"os-R\'enyi graph, the hierarchical incoherence of the graph tends to be high. Moreover the hierarchical incoherence tends to be high when the out-neighbours of its source subgraphs are not clustered together.

\begin{figure}[t]
    \captionsetup{format=plain}
    \centering
    \begin{subfigure}[t]{0.31\textwidth}
    \includegraphics[height=1.6in]{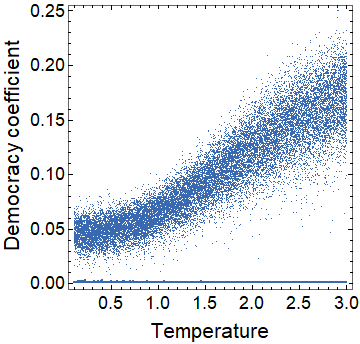}
    \subcaption{}
    \end{subfigure}~
    \begin{subfigure}[t]{0.31\textwidth}
    \includegraphics[height=1.6in]{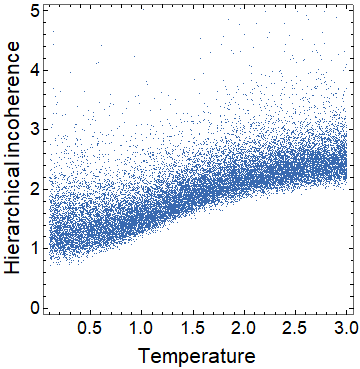}
    \subcaption{}
    \end{subfigure}~
    \begin{subfigure}[t]{0.31\textwidth}
    \includegraphics[height=1.6in]{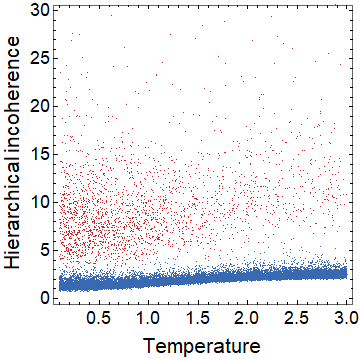}
    \subcaption{}
    \end{subfigure} 
    \caption{Scatter plots of hierarchical incoherence and democracy coefficient over temperature for NSPPM graphs. (a) Democracy coefficient over temperature. The graphs with democracy coefficient less or equal than $20/2500$ form a very tight band on the bottom of the figure. The graphs with democracy coefficient greater than $20/2500$ form a much wider band and there is a clear gap between them. (b) Hierarchical incoherence over temperature of graphs with democracy coefficient greater than $20/2500$. A well defined and relatively narrow band can be seen.  (c) Hierarchical incoherence over temperature of all graphs. For clarity, graphs with democracy coefficient less or equal to $20/2500$ are represented by red points.}
    \label{fig_NSPPM_scatter_plots}
\end{figure}

\section{Contagion Dynamics}
\label{Contagion}

\begin{figure}[t]
    \centering
    \begin{subfigure}[t]{0.48\textwidth}
    {\includegraphics[height=2.2in]{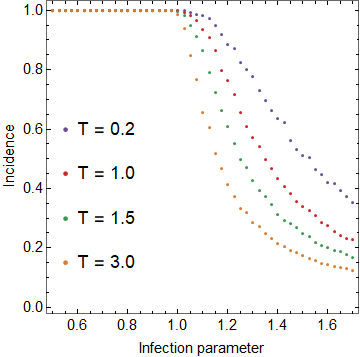}}
    \subcaption{}
    \label{fig_varying_T_A_Incidence_mean_scatter_plot}
    \end{subfigure}~
    \begin{subfigure}[t]{0.48\textwidth}
    {\includegraphics[height=2.2in]{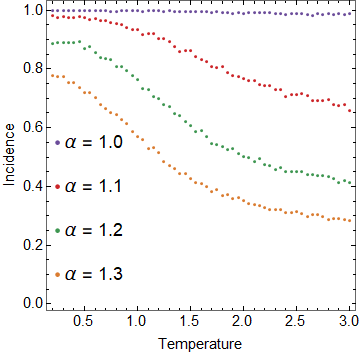}}
    \subcaption{}
    \label{fig_T_varying_A_Incidence_mean_scatter_plot}
    \end{subfigure}
    \caption{Scatter plot of average incidence values from Monte Carlo simulations of the infection spreading with varying temperature $T$ and infection parameter $\alpha$. The average is taken over $1000$ runs. (a) Incidence against $\alpha$ for different values of $T$. (b) Incidence against $T$ for different values of $\alpha$.}
    \label{fig_T_A_Incidence_mean_scatter_plot}
\end{figure}

For contagion dynamics we used a simple Susceptible-Infected-Susceptible epidemic model \cite{pastor2015epidemic}. Even though the model is not very realistic, its simplicity makes it relatively popular. Our aim was to show that the hierarchical structure of a graph can give us dynamical insight, so this model was deemed sufficient. We assume that each vertex has two states, \textit{susceptible} and \textit{infected}. Following \cite{klaise2016neurons}, the probability that vertex $i$ is infected at time $t+1$ is
\begin{align*}
\mathbb P(i\text{ is infected at time }t+1) = f_i(t)^\alpha,
\end{align*}

where $f_i(t)$ is the fraction of $i$'s in-neighbours which are infected at time $t$ and $\alpha$ is a positive parameter that controls the infection rate. Notice that the probability does not depend on the state of $i$ at all. The parameter $\alpha$ is used to tune how infectious the ``disease'' is. In our simulation we use $\alpha$ in the range $[0.5,1.7]$ with step $0.05$.  A small $\alpha$ means it easy for a vertex to get infected and a large $\alpha$ means it hard for a vertex to get infected. For example when $\alpha = 0.5$, a vertex has probability at least $50\%$ to be infected if at least a quarter of its in-neighbours are infected. On the other hand, when $\alpha = 1.7$, a vertex has probability at least $50\%$ to be infected if at least two thirds of its in-neighbours are infected. 

We generated graphs using NSPPM with $N=500$, $B=25$ and $E=2500$ for temperatures in the range $[0.2,3]$ with step $0.025$. We created 1000 graphs for each set of parameters and computed the democracy coefficient and the hierarchical incoherence of each. Then for each values of $\alpha$ we infected the 25 vertices with the lowest hierarchical level and noted the incidence 1000 times. The simulation continued until the incidence became 1 or no vertex was infected or it reached time step 1000.

We found that on graphs with small democracy coefficient, of the first type, the simulation tended to either end quickly with incidence 1 or time out. This is due to the fact that when a graph has a small forward influencing subgraph, all the vertices in the subgraph tend to have low hierarchical levels. This means that the subgraph starts infected, it stays infected for ever. For small $\alpha$ the infection spreads very quickly everywhere so the simulation exits with incidence 1, but for large $\alpha$ the infection never disappears so the simulation times out. Because of this we have excluded these graphs from the results.

On graphs of the second type the infection behaved differently depending on the hierarchical incoherence. On graphs with high incoherence, the infection spread seemingly randomly. On graphs with low incoherence the infection looked like a wave that travelled through the graph. However since NSPPM graphs are PPM graphs where the source vertices gained an in-neighbour, there can be vertices with low hierarchical level that have only one in-neighbour with relatively high hierarchical level. This means that as the infection wave travels through graph, the in-neighbour becomes infected and since it is the only in-neighbour the vertex become infected. This creates another wave that travels through the graph. So in graphs with democracy coefficient higher than $20/2500$ and low hierarchical incoherence, close to 1, we usually found that the infection spread through the graph in periodically generated waves.

Figure \ref{fig_T_A_Incidence_mean_scatter_plot} contains scatter plots of incidence against temperature and infection parameter. Graphs of the first type were discarded as well as ones in which the simulations timed out. The average was taken out of 1000 non timed out simulations. We see that when $\alpha$ is 1 or smaller the incidence is practically 1. Once $\alpha$ becomes larger than 1, then $T$ starts playing a role. We see that higher temperature means lower incidence. Scatter plots of incidence against incoherence and infection parameter can be seen in Figure \ref{fig_Q_A_Incidence_mean_scatter_plot}. Incoherence varied between the interval $(0.8,3.8)$. Since the incoherence of a graph cannot be chosen, each average was taken over graphs with incoherence in a small interval. Because there were fewer graphs towards the edges of the interval, the error is larger there. This is visible on the graph. As expected higher incoherence leads to lower incidence.
Heat maps of incidence can be seen in Figure \ref{fig_T_Q_A_Incidence_mean_heat_map}. The first figure is the heat map of incidence against infection parameter and temperature
and the second figure is the heat map of incidence against infection parameter and hierarchical incoherence.

\begin{figure}[t]
    \centering
    \begin{subfigure}[b]{0.45\textwidth}
    \includegraphics[height=2.6in]{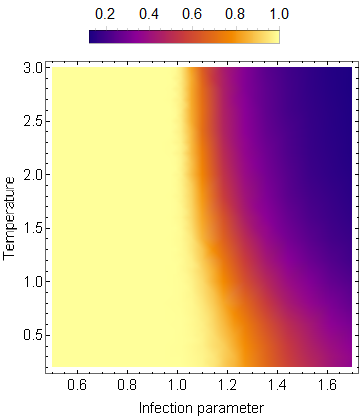}
    \subcaption{}
    \label{fig_T_Q_A_Incidence_mean_heat_map_T_A}
    \end{subfigure}~
    \begin{subfigure}[b]{0.45\textwidth}
    \includegraphics[height=2.6in]{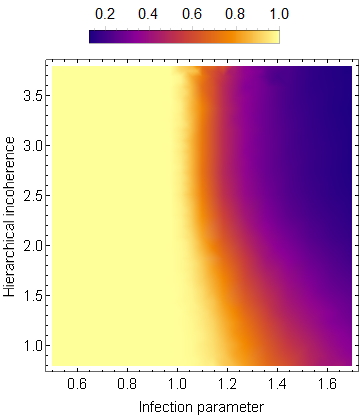}
    \subcaption{}
    \label{fig_T_Q_A_Incidence_mean_heat_map_Q_A}
    \end{subfigure}
    \caption{Heat map of average incidence values from Monte Carlo simulations of the infection spreading. (a) Incidence against $\alpha$ and $T$. (b) Incidence against $\alpha$ and $\fhierarchicalincoherence{G}$.}
    \label{fig_T_Q_A_Incidence_mean_heat_map}
\end{figure}

\end{document}